%% file: conference_101719.tex
\documentclass[conference,onecolumn,draftclsnofoot,12pt]{IEEEtran}

\IEEEoverridecommandlockouts
\usepackage{cite}
\usepackage{amsmath,amssymb,amsfonts}
\usepackage{algorithm}
\usepackage{algpseudocode}
\usepackage{graphicx}
\usepackage{textcomp}
\usepackage[table]{xcolor}
\usepackage{acronym}
\usepackage{bm}
\usepackage{hyperref}
\usepackage{tikz}
\usetikzlibrary{arrows.meta, positioning}
\usepackage{multirow}
\usepackage{caption}
\usepackage{subcaption}
\def\BibTeX{{\rm B\kern-.05em{\sc i\kern-.025em b}\kern-.08em
    T\kern-.1667em\lower.7ex\hbox{E}\kern-.125emX}}

\input{abbreviations}
\input{symbols}

\usepackage{amsthm}

\newtheoremstyle{mystyle}%
  {}%
  {}%
  {\itshape}%
  {}%
  {\bfseries}%
  {.}%
  { }%
  {\thmname{#1}\thmnumber{ #2}\thmnote{ (#3)}}%

\theoremstyle{mystyle}
\newtheorem{proposition}{Proposition}

\newtheorem{corollary}{Corollary}[proposition]

\newtheorem{remark}{Remark}

\begin{document}

\newcommand{\scaleSection}{\vspace{-0.205cm}}
\newcommand{\scaleSubsection}{\vspace{-0.155cm}}
\newcommand{\scaleSubsubsection}{\vspace{-0cm}}
\newcommand{\scaleSectionBelow}{\vspace{-0.175cm}}
\newcommand{\scaleSubsectionBelow}{\vspace{-0.175cm}}
\newcommand{\scaleSubsubsectionBelow}{\vspace{-0.025cm}}
\newcommand{\scaleAlign}{\vspace{-0.185cm}}

\title{Molecule Mixture Detection and Design for MC Systems with Non-linear, Cross-reactive Receiver Arrays\vspace*{-15mm}}
\author{\IEEEauthorblockN{Bastian Heinlein, Kaikai Zhu, Sümeyye Carkit-Yilmaz, Sebastian Lotter,\\ Helene M. Loos, Andrea Buettner, Yansha Deng, Robert Schober, and Vahid Jamali}
\thanks{
Bastian Heinlein, Kaikai Zhu, Sümeyye Carkit-Yilmaz, Sebastian Lotter, Helene M. Loos, Andrea Buettner, and Robert Schober are with Friedrich-Alexander-Universität Erlangen-Nürnberg, Erlangen, Germany. Bastian Heinlein and Vahid Jamali are with Technical University of Darmstadt, Darmstadt, Germany. Yansha Deng is with King's College, London, United Kingdom. Helene M. Loos and Andrea Büttner are also with the Fraunhofer Institute for Process Engineering and Packaging IVV, Freising, Germany. Helene M. Loos is also with Rheinische Friedrich-Wilhelms-Universität Bonn, Bonn, Germany.}
\thanks{This work was funded by the Deutsche Forschungsgemeinschaft (DFG, German Research Foundation) – GRK 2950 – Project-ID 509922606.}
\thanks{This work has been presented in part at the ACM Nanoscale Computing and Communications Conference, 2025,~\cite{heinlein:nanocom}.}
}

\maketitle
\vspace*{-20mm}
\begin{abstract}
\vspace*{-5mm}
Air-based \ac{MC} has the potential to be one of the first \ac{MC} systems to be deployed in real-world applications, enabled by commercially available sensors. However, these sensors usually exhibit non-linear and cross-reactive behavior, contrary to the idealizing assumption of linear and perfectly molecule type-specific sensing often made in the \ac{MC} literature.  
To address this mismatch, we propose several detectors and transmission schemes for a molecule mixture communication system where the \ac{RX} employs non-linear, cross-reactive sensors.
All proposed schemes are based on the first- and second-order moments of the symbol likelihoods that are fed through the non-linear \ac{RX} using the Unscented Transform.
In particular, we propose an approximate maximum likelihood (AML) symbol-by-symbol detector for \ac{ISI}-free transmission scenarios and a complementary mixture alphabet design algorithm which accounts for the \ac{RX} characteristics.
When significant \ac{ISI} is present at high data rates, the AML detector can be adapted to exploit statistical \ac{ISI} knowledge. Additionally, we propose a sequence detector which combines information from multiple symbol intervals.
For settings where sequence detection is not possible due to extremely limited computational power at the \ac{RX}, we propose an adaptive transmission scheme which can be combined with symbol-by-symbol detection.
Using computer simulations, we validate all proposed detectors and algorithms based on the responses of commercially available sensors as well as artificially generated sensor data incorporating the characteristics of \acl{MOS} sensors.
By employing a general system model that accounts for transmitter noise, \ac{ISI}, and general non-linear, cross-reactive \ac{RX} arrays, this work enables reliable communication for a large class of \ac{MC} systems.
\end{abstract}

\acresetall
\renewcommand{\baselinestretch}{1.2}
\input{sections/introduction}
\input{sections/system_model}

\input{sections/detector}
\input{sections/mixture}
\input{sections/evaluation}

\input{sections/conclusion}
\input{sections/appendix}

\vspace*{-3mm}
\bibliographystyle{ieeetr}
\bibliography{literature}
\clearpage
\end{document}

%% file: abbreviations.tex
\acrodef{MC}{molecular communication}
\acrodef{ML}{maximum likelihood}
\acrodef{AML}{approximate maximum likelihood}
\acrodef{EM}{electro-magnetic}
\acrodef{MOS}{metal-oxide semi-conductor}
\acrodef{CSK}{concentration shift keying}
\acrodef{RSK}{ratio shift keying}
\acrodef{MSK}{molecule shift keying}
\acrodef{GMoSK}{generalized molecule shift keying}
\acrodef{BCSK}{binary concentration shift keying}
\acrodef{MMSK}{molecule mixture shift keying}
\acrodef{MIMO}{multiple-input multiple-output}
\acrodef{UT}{Unscented Transform}
\acrodef{RX}{receiver}
\acrodef{TX}{transmitter}
\acrodef{RV}{random variable}
\acrodef{pdf}{probability density function}
\acrodef{ISI}{inter-symbol interference}
\acrodef{SIN}{signal-independent noise}
\acrodef{SDCN}{signal-dependent channel noise}
\acrodef{BER}{bit error rate}
\acrodef{SER}{symbol error rate}
\acrodef{SNR}{signal-to-noise ratio}
\acrodef{kNN}{$k$-nearest-neighbor}
\acrodef{PEP}{pairwise error probability}
\acrodef{MDA}{mixture design algorithm}
\acrodef{SOD}{sensor output domain}
\acrodef{SID}{sensor input domain}

\acrodef{pdf}{probability density function}
\acrodef{CSK}{concentration shift keying}
\acrodef{RX}{receiver}
\acrodef{TX}{transmitter}
\acrodef{MOS}{metal-oxide semiconductor}
\acrodef{MMSK}{molecule mixture shift keying}
\acrodef{MC}{molecular communication}
\acrodef{UT}{unscented transform}
\acrodef{SER}{symbol error rate}

\acrodef{LC}{low complexity}
\acrodef{GMM}{Gaussian Mixture Model}

\acrodef{ODD}{odor delivery device}
\acrodef{VOC}{volatile organic compounds}

%% file: symbols.tex
\newcommand{\kappamax}{\kappa^{\max}}

\renewcommand{\a}{\mathbf{a}}
\renewcommand{\b}{\mathbf{b}}

\newcommand{\x}{\mathbf{x}}

\newcommand{\xbar}{\mathbf{\bar{x}}}

\renewcommand{\H}{\mathbf{H}}

\newcommand{\y}{\mathbf{y}}
\newcommand{\z}{\mathbf{z}}
\newcommand{\nbase}{\bar{\mathbf{n}}}
\newcommand{\ntx}{\mathbf{n}^{\mathrm{TX}}}
\newcommand{\nc}{\mathbf{n}^{\mathrm{C}}}
\newcommand{\nrx}{\mathbf{n}^{\mathrm{RX}}}
\newcommand{\f}[1]{\mathbf{f}\left( #1 \right)}

\newcommand{\g}[1]{\mathbf{g}\left( #1 \right)}
\newcommand{\nuc}{\nu^{\mathrm{C}}}
\newcommand{\n}{\mathbf{n}}

\newcommand{\E}[1]{\mathrm{E}\left\{ #1 \right\}}
\newcommand{\Cov}[1]{\mathrm{Cov}\left\{ #1 \right\}}
\newcommand{\EUT}[1]{\mathrm{E}_{\mathrm{UT}}\left\{ #1 \right\}}
\newcommand{\CovUT}[1]{\mathrm{Cov}_{\mathrm{UT}}\left\{ #1 \right\}}

\newcommand{\symbolalphabet}{\mathcal{A}}
\newcommand{\feasibleset}{\mathcal{F}_{\x}}
\newcommand{\feasiblesety}{\mathcal{F}_{\y}}

\newcommand{\nsymbols}{N}
\newcommand{\nspecies}{S}
\newcommand{\nsensors}{R}
\newcommand{\nsigmapoints}{n_{\sigma}}

\newcommand{\I}{\mathbf{I}}
\newcommand{\nullmatrix}{\mathbf{0}}

\newcommand{\SIN}{^{\mathrm{SIN}}}
\newcommand{\SDCN}{^{\mathrm{SDCN}}}

\newcommand{\deuclidean}{d^{\ell_2}}
\newcommand{\dvahid}{d^{\mathrm{SNR}}}

\newcommand{\s}{\mathbf{s}}

\newcommand{\transpose}{^\mathrm{T}}
\newcommand{\hadamard}{\odot}

\newcommand{\muvec}{\boldsymbol{\mu}}
\newcommand{\muvechat}{\boldsymbol{\hat{\mu}}}
\newcommand{\covmat}{\mathbf{C}}
\newcommand{\covmathat}{\mathbf{\hat{C}}}

\newcommand{\shatlc}{\hat{s}^{\mathrm{LC}}}

\newcommand{\W}{\mathbf{W}}

\renewcommand{\c}{\mathbf{c}}
\newcommand{\candidateset}{\mathcal{C}}

\newcommand{\ybar}{\bar{\y}}

%% file: sections/introduction.tex
\scaleSection\section{Introduction}\label{sec:introduction}\scaleSectionBelow
\Ac{MC} constitutes a new paradigm in communication engineering by employing molecules for information transmission in scenarios where \ac{EM} wave-based communication is not practical such as for interfacing with biological entities via chemical signals or in \ac{EM} wave-denied environments \cite{guo:molecular_physical_layer_6g}. 
While most work in \ac{MC} has considered scenarios where information is encoded into the concentration of a single molecule type, i.e., \ac{CSK}, modulation schemes employing multiple molecules offer additional degrees of freedom. 
The arguably simplest example of such a multi-molecule scheme is \ac{MSK} where the type of molecule that is released in a symbol interval indicates the transmitted symbol, promising higher data rates compared to \ac{CSK}~\cite{eckford:achievable_information_rates_mc_distinct_molecules}. Based on~\cite{eckford:achievable_information_rates_mc_distinct_molecules}, a variety of more complex modulation schemes employing multiple molecule types have been proposed to further aid \ac{ISI} mitigation, deal with time-varying channels, and increase data rate, e.g., ~\cite{chen:gmosk,araz:rskm_time_varying_MC_channels,kilic:mrsk,gang:molecular_type_permuation_shift_keying_mc}.

However, most existing literature on multi-molecule modulation schemes assumes direct access to the concentration $c_s$ of each molecule type $s \in \{1,\dots,\nspecies\}$ through a molecule counting receiver whose response $r_s$ to the $s$-th molecule type is linearly proportional to $c_s$, i.e., $r_s \sim c_s$. Here, $\nspecies$ denotes the total number of considered molecule types. 
While the concentrations of individual molecule types can in principle be measured in macro-scale applications, e.g., using proton-transfer-reaction mass spectrometry or membrane inlet mass spectrometry, these devices are prohibitively expensive for most applications and are not suitable for mass deployment~\cite{mcguiness:experimental_results_openair_transmission_macromolecular_communication}. On the other hand, cheaper sensors exist, such as optical sensors for liquid-based \ac{MC} or \ac{MOS} sensors for air-based \ac{MC}, but these sensors are usually cross-reactive or non-linear or both.
For example, molecules have often overlapping absorption spectra, rendering optical sensors \textbf{cross-reactive}~\cite{wietfeld:evaluation_multi_molecule_molecular_communication_testbed,walter:cnnbased_detection_mixed_molecule_concentrations}, i.e., the sensor response $r$ at a particular wavelength is given by $r = w_1 \cdot c_1 + \dots + w_{\nspecies} \cdot c_{\nspecies}$, where $w_s$ represents the weighting factor for the $s$-th molecule type. 
\ac{MOS} sensors, on the other hand, are well-known to be \textbf{non-linear}. Specifically, they exhibit power-law behavior, i.e., $r = a\cdot c^b$, where $a,b \in \mathbb{R}$~\cite{yamazoe:theory_power_laws_semiconducator_gas_sensors}. In practice, \ac{MOS} sensors also react to multiple molecule types at once~\cite{albert:crossreactive_chemical_sensor_arrays}, i.e., they are both non-linear and cross-reactive at the same time. If a sensor is both cross-reactive and non-linear, its response is described by a non-linear function $f(c_1, \dots, c_{\nspecies})$. 

Because experimental work in \ac{MC} usually relies on such non-linear or cross-reactive sensors, the experimentally implemented modulation schemes have been comparatively simple so far: While the use of optical sensors for liquid-based multi-molecule communication  has been investigated in~\cite{walter:cnnbased_detection_mixed_molecule_concentrations,wietfeld:evaluation_multi_molecule_molecular_communication_testbed,cali:experimental_implementation_molecule_shift_keying}, these testbeds are limited to a small number of molecule types and the molecule alphabets were chosen ad hoc, neglecting opportunity for improved system performance by optimized molecule alphabets. \ac{MOS} sensors, on the other hand, have been employed mostly for \ac{CSK} transmissions, also with ad hoc chosen alphabets, e.g., \cite{kim:experimentally_validated_channel_model_for_mc_systems,farsad:tabletop_mc_text_messages_through_chemical_signals}. The work in~\cite{wisayataksin:4ary_odor_shift_keying} is a notable exception among the air-borne testbeds employing \ac{MOS} sensors: Here, information is transmitted by releasing either no molecules, ethanol, acetone, or both and recovered by two \ac{MOS} sensors at the \ac{RX}. However, this work does not investigate higher-order modulation schemes and relies only on an ad hoc chosen alphabet. 

We attribute the absence of low-cost practical \ac{MC} systems employing higher-order multi-molecule modulation schemes and non-linear, cross-reactive \acp{RX} in part to a lack of a corresponding theoretical foundation. The primary work to theoretically investigate a macro-scale \ac{RX} that is both cross-reactive and non-linear is \cite{jamali:olfaction_inspired_MC}. There, a simple receptor array leveraging compressive sensing for symbol detection, inspired by the olfactory system of insects, is investigated in an \ac{ISI}-free setting. The information is encoded into the concentration of multiple molecule types at once (i.e., into a \textbf{molecule mixture}), with the different mixtures being optimized using a greedy algorithm to elicit different responses of the receptor array.  
However, to fully leverage the potential of multi-molecule communication for higher data rates, \ac{ISI} has to be considered. Yet, so far, approaches for \ac{ISI}-mitigation at the \ac{TX}, e.g., via channel coding or by adaptively choosing the amount of molecules transmitted, have been mostly limited to \ac{CSK}~\cite{tepekule:isi_mitigation_techniques_mc,jing:lightweight_channel_code_isi_mitigation_mc,arjmandi:isi_avoiding_modulation_diffusion_based_mc}. On the other hand, schemes for \ac{ISI} mitigation at the \ac{RX} typically assume direct access to molecule concentrations and are thus not applicable in settings with non-linear, cross-reactive sensors~\cite{tepekule:isi_mitigation_techniques_mc,jamali:design_mf_molecule_counting_rxs,kuran:survey_modulation_techniques_mcd}. 

\scaleSubsection\subsection{Main Contributions}\scaleSubsectionBelow\label{sec:introduction:contributions}
Motivated by the lack of study on molecule mixture communication with non-linear, cross reactive \ac{RX} arrays, especially for \ac{ISI} channels, we propose methods to enable reliable communication in such scenarios. In particular, our key contributions are summarized as follows.
\begin{itemize}
    \item We introduce a \textbf{comprehensive system model} for molecule mixture communication that accounts for various experimentally identified noise sources (i.e., \ac{TX}, channel, and \ac{RX} noise) impairing communication performance and an \ac{RX} employing a non-linear, cross-reactive sensor array. Unlike the conference version of this paper~\cite{heinlein:nanocom}, the model accounts for \ac{ISI}, thus making it applicable to settings with shorter symbol intervals, as needed to achieve higher data rates.
    \item To perform detection for this general system model, we employ the approximate likelihoods of the transmitted mixtures based on the \textbf{first- and second-order moments} of the \ac{RX}'s sensor response. We derive the relevant moments by propagating the uncertainty introduced by the noise sources through the system model leveraging the \ac{UT}. 
    \item In the \textbf{\ac{ISI}-free case}, we leverage these moments to derive an \ac{AML} symbol-by-symbol detector and a \textit{static}\footnote{We refer to this algorithm as \textit{static} because it selects an alphabet of molecule mixtures once before the transmission starts and then uses this alphabet for all symbol intervals.} \ac{MDA}, which optimizes mixtures so that the sensor outputs at the \ac{RX} are well-separable.
    \item Extending \cite{heinlein:nanocom}, we also leverage the derived moments to derive several methods to account for \ac{ISI}: We adapt the \ac{AML} detector to \textbf{take into account statistical information about the \ac{ISI}}. This symbol-by-symbol detector in conjunction with the previously mentioned \ac{MDA} is especially suitable for settings where both \ac{RX} and \ac{TX} have extremely limited computational capabilities. We also propose a sequence detector to enable lower \acp{SER} when more computational power is available at the \ac{RX}. Finally, we propose an adaptive transmission scheme for settings with a computationally more powerful \ac{TX}, enabling similar \acp{SER} as achieved by the sequence detector while requiring only a symbol-by-symbol detector at the \ac{RX}. 
\end{itemize}

\scaleSubsection\subsection{Structure and Notation}\scaleSubsectionBelow\label{sec:introduction:structure_notation}
The remainder of this paper is structured as follows: Section~\ref{sec:system_model} introduces the proposed system model. Section~\ref{sec:detection} then presents the various moment-based detectors. This is followed by the derivation of the \ac{MDA} and the adaptive transmission scheme in Section~\ref{sec:mixture}. Finally, we evaluate the performance of the proposed schemes in Section~\ref{sec:evaluation} before concluding the paper in Section~\ref{sec:conclusion}. 

\textbf{Notation}: Vectors and matrices are denoted by lowercase and uppercase bold letters, respectively. The $i$-th entry of vector $\x$ is denoted by $[\x]_i$ and the $j$-th column of matrix $\covmat$ as $[\covmat]_j$. The transpose of matrix $\covmat$ is denoted by $\covmat\transpose$ whereas its determinant is denoted by $\det\{\covmat\}$. 
Element-wise multiplication of two vectors $\x$ and $\y$ is shown by $\x\hadamard\y$. The \ac{pdf} of a \ac{RV} $x$ conditioned on $y$ is given by $p_{x|y}(x|y)$. The expectation of a \ac{RV} is denoted by $\E{\cdot}$ and the covariance matrix of a random vector by $\Cov{\cdot}$. $\I$ denotes the identity matrix, $\mathbf{1}$ denotes an all-one column vector, and $\nullmatrix$, depending on the context, either a column vector or a matrix containing only zeroes. $\lVert \cdot \rVert_2$ denotes the $\ell_2$-norm. 
Sets are shown by calligraphic letters and the cardinality of set $\mathcal{A}$ is denoted by $|\mathcal{A}|$. $\mathcal{U}(\mathcal{A})$ denotes a \ac{RV} that is uniformly distributed over a \mbox{set $\mathcal{A}$}. $\mathbb{N}$ and $\mathbb{R}$ denote respectively the sets of real and natural numbers. Analogously, and $\mathbb{R}_{\geq0}$ and $\mathbb{N}_{\geq0}$ denote the sets of non-negative real and the non-negative natural numbers, respectively.

%% file: sections/system_model.tex
\scaleSection\section{System Model}\scaleSectionBelow\label{sec:system_model}
In the following, we first introduce the system model in its general form in Section~\ref{sec:system_model:overview} before characterizing the response of \ac{MOS} sensors to molecule mixtures and relevant noise distributions in Sections~\ref{sec:system_model:rx_characteristics} and~\ref{sec:system_model:noise_distributions}, respectively. 

\scaleSubsection\subsection{Overview}\scaleSubsectionBelow\label{sec:system_model:overview}
\begin{figure*}
    \vspace*{-10mm}
    \centering\includegraphics[width=\textwidth]{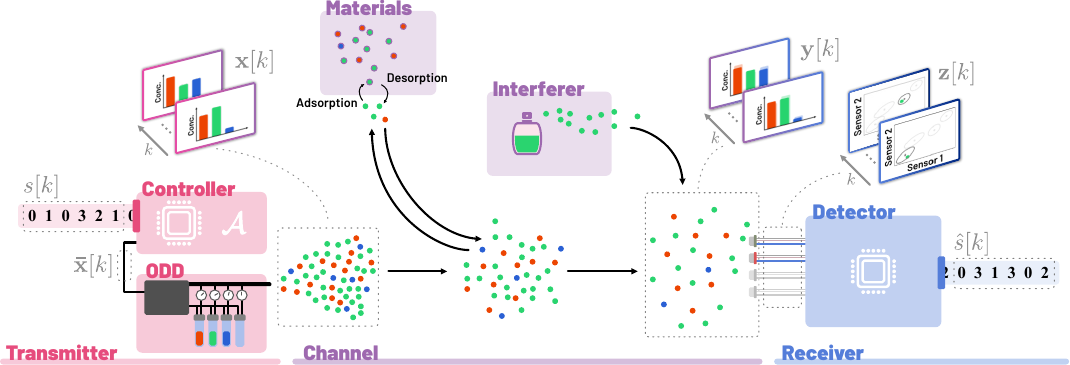}
    \vspace*{-9mm}
    \caption{\textbf{System overview.} At the \ac{TX}, a controller converts a sequence of symbols $s[k]$ to control signals for the \acl{ODD} (ODD), which should release molecule mixtures $\xbar[k]$, but releases mixtures $\x[k]$ due to hardware imperfections instead. At the \ac{RX}, molecule mixtures $\y[k]$ are observed because $\x[k]$ is subject to attenuation, e.g., due to dispersion or adsorptions to materials, and interferers might introduce additional molecules. The output of the \ac{RX} sensor array, $\z[k]$, is then used by the detector to derive a symbol estimate $\hat{s}[k]$.}  \label{fig:system_model:overview}
    \vspace*{-10mm}
\end{figure*}
In this work, we consider the general system model illustrated in Figure~\ref{fig:system_model:overview} that accounts for the three major noise sources reported in the experimental \ac{MC} literature, namely release noise~\cite{farsad:tabletop_mc_text_messages_through_chemical_signals,kim:experimentally_validated_channel_model_for_mc_systems}, channel noise~\cite{kim:experimentally_validated_channel_model_for_mc_systems}, and the noise of the sensor array~\cite{shin:low_frequency_noise_gas_sensors,farsad:tabletop_mc_text_messages_through_chemical_signals}. It also accounts for \ac{ISI} in the propagation channel and the non-linear, cross-reactive behavior of the sensor array at the \ac{RX}. To limit complexity, we adopt a single-sample detector, i.e., one sample is taken per symbol interval. 
Thus, we arrive at the following system model:
\allowdisplaybreaks
\begin{subequations}
    \begin{align}
        \x[k] &= \xbar[k] +\ntx(\xbar[k]) \label{eq:system_model:tx_noise} \\
        \ybar[k] &= \sum_{\kappa=0}^{\kappamax} \H[\kappa] \x[k-\kappa]\label{eq:system_model:isi} \\
        \y[k] &= \ybar[k] + \nc\left(\ybar[k]\right) \label{eq:system_model:channel_noise}\\
        \z[k] &= \f{\y[k]} + \nrx(\f{\y[k]}).\label{eq:system_model:rx_noise}
    \end{align}
\end{subequations}
Eq.~\eqref{eq:system_model:tx_noise} describes the noisy release process of the signaling molecules, while~\eqref{eq:system_model:isi} models how the current and previous symbol transmissions accumulate in the propagation channel, and \eqref{eq:system_model:channel_noise} accounts for the impact of propagation uncertainty and interference. Finally, \eqref{eq:system_model:rx_noise} describes how the measurements, which are used for symbol detection, relate to the molecule concentrations observed at the \ac{RX}. The variables appearing in \eqref{eq:system_model:tx_noise}-\eqref{eq:system_model:rx_noise} will be explained in the following.

We consider a molecule mixture communication system with mixture alphabet $\symbolalphabet$ of size $\nsymbols$. The controller at the \ac{TX} receives the index $s[k] \in \mathcal{S}_0 = \{1, \dots, \nsymbols\}$ of the mixture that is to be transmitted in the $k$-th interval and then determines how many molecules of each type should be released, summarized in $\xbar[k] \in \mathbb{R}_{\geq 0}^{\nspecies}$ , where $\nspecies$ denotes the number of molecule types available at the \ac{TX}. The controller then informs the release hardware, such as an \ac{ODD}, about the desired release amount $\xbar[k]$~\cite{hopper:multichannel_portable_odor_delivery_device}. 
Eq.~\eqref{eq:system_model:tx_noise} represents that real hardware cannot perfectly control the number of molecules it releases~\cite{farsad:tabletop_mc_text_messages_through_chemical_signals}, resulting in a \textbf{released molecule concentration} $\x[k]$ instead of $\xbar[k]$. We model this release noise by additive noise $\ntx(\xbar[k])$ which is memory-free but can depend on the desired mixture $\xbar[k]$ and follows the conditional \ac{pdf} $p_{\ntx|\xbar[k]}(\ntx)$. It is noted that the number of released molecules cannot be arbitrarily large, but is limited, e.g., by hardware constraints. Therefore, we require that all symbols $\xbar[k] \in \symbolalphabet$ lie in a \textit{feasible set} $\feasibleset \subseteq \mathbb{R}_{\geq 0}^{\nspecies}$. In this work, we consider $\feasibleset$ to be a hypercube, i.e.,  there is an individual constraint on each molecule type. However, our proposed schemes are also applicable to other shapes of the feasible set, e.g., when limiting the total number of released molecules.

Eq.~\eqref{eq:system_model:isi} models the expected molecule concentrations at the \ac{RX}, $\ybar[k] \in \mathbb{R}_{\geq 0}^{\nspecies}$ after \textbf{propagation through the \ac{ISI} channel}. The fraction of the released molecules expected at the \ac{RX} is characterized by diagonal channel matrix $\H[\kappa]$. Here, $[\H[\kappa]]_{ii}$ denotes the attenuation factor for molecules of type $i$ released $\kappa\!\! \in \!\!\{0, \dots, \kappamax\}\!$ symbol intervals ago. The specific entries in $\H[\kappa]$ depend, among other things, on the dispersion in the environment, degradation (e.g., due to photochemical oxidation), and adsorption to and desorption from materials in the environment~\cite{carkit:experimental_analysis_surface_induced_molecular_adsorption,farsad:tabletop_mc_text_messages_through_chemical_signals,williams:human_odour_thresholds_tuned_atmospheric_chemical_lifetimes}. Here, we assume that the molecules propagate independently from each other.

Eq.~\eqref{eq:system_model:channel_noise} captures how $\ybar[k]$ is affected by \textbf{interference or propagation uncertainty} (e.g., each sensor at the \ac{RX} might observe slightly different molecule concentrations due to the spatial separation of the sensors and diffusion noise), resulting in molecule concentrations $\y[k]$ at the \ac{RX}. We model the impact of these uncertainties by channel noise $\nc(\ybar[k])$ which can be signal-dependent and is distributed according to a conditional \ac{pdf} $p_{\nc|\y[k]}(\nc)$. Throughout this paper, we will use the term \ac{SID} when referring to the moments or observations of the concentrations at the \ac{RX}.

Finally,~\eqref{eq:system_model:rx_noise} describes how the molecule concentrations $\y[k]$ are converted to the \textbf{response of the sensing units}, $\z[k] \in \mathbb{R}^{\nsensors}$, where $[\z[k]]_r$, $r \in \{1,\dots,R\}$, denotes the output of the $r$-th sensing unit and $R$ is the total number of sensing units at the \ac{RX}. The detector uses $\z[k]$ to estimate the transmitted symbol as $\hat{s}[k]$. The mapping from $\y[k]$ to $\z[k]$ follows a non-linear function $\f\cdot$, whose output depends on multiple molecule concentrations simultaneously. We model the sensing unit as memory-free, which is realistic if the propagation dynamics are slow compared to the sensor dynamics. This is true for various types of practical sensors, including but not limited to optical sensors and many \ac{MOS} sensors~\cite{dennler:high_speed_odor_sensing,wu:research_progress_mems_gas_sensors,wietfeld:evaluation_multi_molecule_molecular_communication_testbed}. 
In practice, the sensing units also introduce measurement noise $\nrx(\f{\y[k]})$, e.g., thermal noise or quantization noise. For the sake of generality, we again permit signal-dependent measurement noise and characterize it by a conditional \ac{pdf} $p_{\nrx|\f{\y[k]}}(\nrx)$. Throughout this paper, we will use the term \ac{SOD} when referring to moments or observations related to~$\z[k]$.

In summary, we consider a very general system model accounting for the relevant noise sources reported in the experimental \ac{MC} literature, channel \ac{ISI}, and a possibly non-linear sensor function. 

\scaleSubsection\subsection{MOS Sensors as a Non-linear, Cross-reactive RX Array}\scaleSubsectionBelow\label{sec:system_model:rx_characteristics}
In this work, we consider $\nsensors$ \ac{MOS} sensors as an example of an array of non-linear, cross-reactive sensing units. As already discussed in Section~\ref{sec:introduction}, it is well established that these sensors exhibit power-law behavior, i.e., their conductance $G$ for a given concentration $c$ of a target gas can be described by $G = a \cdot c^b$, where $a,b \in \mathbb{R}$ are constants specific to the considered combination of sensor and target gas~\cite{yamazoe:theory_power_laws_semiconducator_gas_sensors}. However, there is no clear consensus on the function that describes the response to gas mixtures. Instead, a variety of models have been proposed in the literature, e.g.,~\cite{llobet:steadystate_transient_behavior_of_thickfilm_tin_oxide_sensors_gas_mixtures,madrolle:linear_quadratic_model_quantification_mixture_two_diluted_gases_single_MOS,hirobayashi:verification_logarithmic_model,gurin:gas_mixture_estimation_using_powerlaw_models_arrayed_chemiresistive_MOS}. Here, we focus on the model originally introduced in \cite{llobet:steadystate_transient_behavior_of_thickfilm_tin_oxide_sensors_gas_mixtures} since it is applicable to mixtures with more than two components and it was shown to provide a good fit to measurement data in \cite{madrolle:linear_quadratic_model_quantification_mixture_two_diluted_gases_single_MOS}.
For the $r$-th sensor, we have~\cite{llobet:steadystate_transient_behavior_of_thickfilm_tin_oxide_sensors_gas_mixtures}
\begin{equation}\label{eq:system_model:reciever_characteristics:f}
    [\f{\y}]_r = \mathbf{a}_r\transpose \y_{\mathrm{p}} - \y_{\mathrm{p}} \transpose \mathbf{A}_r\y_{\mathrm{p}}.
\end{equation}
Here, $\y_{\mathrm{p}}$ is a $\nspecies \times 1$ vector with entries $[\y_{\mathrm{p}}]_s = [\y]_s^{[\mathbf{b}_r]_s}$, where $s \in \{1,\dots,\nspecies\}$ and $\mathbf{b}_r$ collects the power-law coefficients of the $r$-th sensor for each gas. The remaining sensor parameters are captured in $\a_r$ and $\mathbf{A}_r$, where $[\mathbf{a}_r]_i \geq 0$ $\forall i$ and $[\mathbf{A}_r]_{ij}\geq0$, $\forall i>j$, and $[\mathbf{A}_r]_{ij}=0$ otherwise.

\begin{remark}
Our proposed algorithms are applicable to arbitrary functions $\f\cdot$ and are not limited to this specific model for \ac{MOS} sensors or to \ac{MOS} sensors in the first place. In fact, our algorithms do not even require a closed-form model for the sensor behavior: One could apply our algorithms to preprocessed measurements, making the algorithms relevant even for digital sensors like the BME680 or the DHT22 sensors which have been used in~\cite{bhattacharjee:breath_patterns_as_signals,alzubi:macroscale_MC_iot_based_pipeline_inspection}.
\end{remark}

\scaleSubsection\subsection{Relevant Noise Distributions}\scaleSubsectionBelow\label{sec:system_model:noise_distributions}
Next, we specialize the noise distributions for the system model proposed in Section~\ref{sec:system_model:overview} to relevant special cases that are commonly found in \ac{MC}. Specifically, we consider two scenarios. In the first scenario, all three noise sources are \textbf{signal-independent and Gaussian}, permitting a simplified analytical treatment. This scenario will be referred to as \ac{SIN} in the remainder of this paper. Signal-independent \ac{TX} noise, characterized by mean $\muvec_{\ntx}^{\SIN}$ and covariance $\covmat_{\ntx}^{\SIN}$, might occur, e.g., due to small errors when controlling the flow rates of an \ac{ODD}~\cite{hopper:multichannel_portable_odor_delivery_device} or due to small uncontrollable leakage. Signal-independent \ac{RX} noise, characterized by mean $\muvec_{\nrx}^{\SIN}$ and covariance $\covmat_{\nrx}^{\SIN}$, might be due to quantization noise or thermal noise in electronic components~\cite{shin:low_frequency_noise_gas_sensors}. For high molecule concentrations and controlled propagation environments, the observed number of molecules at the RX converges to its expectation (see~\cite[Sec. 2.2.2.4]{jamali:thesis}). The channel noise, characterized by $\muvec_{\nc}^{\SIN}$ and covariance $\covmat_{\nc}^{\SIN}$, may then be used to represent uncertainty about the amount of background molecules in the environment, e.g., due to unknown interference which is independent of the currently transmitted signal. We consider Gaussian noise for simplicity\footnote{We note that Gaussian \ac{TX} and channel noise can in principle lead to negative concentrations, i.e., unphysical behavior. In these cases, we set the values in the simulations manually to zero. These events occur only in rare circumstances for the adopted parameter setting.} since it has been shown to be an accurate approximation in air-based \ac{MC} \cite{kim:experimentally_validated_channel_model_for_mc_systems,mcguiness:experimental_results_openair_transmission_macromolecular_communication} and since Gaussian \acp{RV} are readily parametrized by their mean and covariance. Yet, the proposed detector and mixture alphabet design algorithm are also applicable to other noise distributions, although possibly with some performance degradation since they exploit only mean and covariance.

We consider a second more challenging scenario of \textbf{\ac{SDCN}}, while the \ac{TX} and \ac{RX} noise remain signal-independent and Gaussian (i.e., they are described by $\muvec_{\ntx}^{\SIN}$, $\covmat_{\ntx}^{\SIN}$, $\muvec_{\nrx}^{\SIN}$, and $\covmat_{\nrx}^{\SIN}$ as before). \ac{SDCN} $\nc(\y[k])$, which occurs, e.g., in Poisson channels, is commonly assumed in \ac{MC}~\cite{jamali:olfaction_inspired_MC} and often approximated as 
\begin{equation}\label{eq:system_model:sdcn_definition}
    \nc(\y[k]) = \sqrt{\nuc\y[k]} \hadamard \nbase,
\end{equation}
where square root $\sqrt{\cdot}$ is applied element-wise and $\nbase$ is normally distributed with zero-mean and covariance matrix $\I$. 
We introduce a scaling factor $\nuc$ to control the noise power introduced by the channel, with $\nuc=1$ corresponding to the Gaussian approximation of the Poisson channel.

%% file: sections/detector.tex
\scaleSection\section{Moment-based Detection for Non-Linear Models}\scaleSectionBelow\label{sec:detection}
In this section, we present the different proposed detectors,  which are all based on the first- and second-order moments of the sensor outputs $\z[k]$. To this end, we first introduce the basic idea behind these detectors in Section~\ref{sec:detection:overview} before presenting the detectors themselves in Section~\ref{sec:detection:detectors}. Then, in Section~\ref{sec:detection:ut}, we review how first- and second-order moments can be propagated through non-linear functions using the \ac{UT} and explain how the \ac{UT} can be employed to estimate the characteristics of signal-dependent noise. Finally, in Section~\ref{sec:detection:moments}, we derive the moments required by the detectors. 

\scaleSubsection\subsection{Overview}\scaleSubsectionBelow\label{sec:detection:overview}
Since the proposed detectors are inspired by \ac{ML} detection, we briefly review the optimal ML detector for \ac{ISI}-free channels (i.e., $\kappamax=0$) before outlining the concept behind \ac{AML} detection and how this idea can be extended to \ac{ISI} channels. 
For \textbf{\ac{ML} detection}, the
likelihoods of all symbols $s[k] \in \mathcal{S}_0$ are exploited to make decisions. In particular, \ac{ML} detection is performed according to 
\begin{equation}
    \hat{s}^{\mathrm{ML}}[k] = \arg\max_{s[k] \in \mathcal{S}_0} p_{\z[k]|s[k]}(\z[k]),
\end{equation}
where $p_{\z[k]|s[k]}(\z[k])$ denotes the pdf of $\z[k]$ in the \ac{SOD}
domain given a transmitted symbol $s[k]$. \ac{ML} detection not only yields optimal performance in terms of the \ac{SER} for equiprobable symbols but also does not require inversion of $\f\cdot$, which may not be invertible in the first place. Unfortunately, obtaining the likelihoods in the \ac{SOD} in closed form may not be possible for non-linear, cross-reactive \ac{RX} arrays.

Thus, we instead propose an approach for \textbf{\ac{AML} detection} based on the approximate first- and second-order moments of $p_{\z[k]|s[k]}(\z[k])$. In this case, the estimated moments $\muvechat_{\mathbf{z}[k]|s[k]}$ and $\covmathat_{\mathbf{z}[k]|s[k]}$ are used to to parametrize a \ac{pdf} $\hat{p}_{\z[k]|s[k]}(\z)$ approximating $p_{\z[k]|s[k]}(\z[k])$, so that \ac{AML} detection can be performed according to:
\begin{equation}\label{eq:detection:aml}
    \hat{s}^{\mathrm{AML}}[k] = \arg\max_{s[k] \in \mathcal{S}_0} \hat{p}_{\z[k]|s[k]}(\z[k]).
\end{equation}
We focus on first- and second-order moments because they are easily propagated through linear systems and the \ac{UT} provides an efficient tool to approximately propagate them through non-linear functions, as we will discuss in detail in Section~\ref{sec:detection:ut}. For the approximate \acp{pdf}, we use Gaussian distributions, i.e., $\hat{p}_{\z[k]|s[k]}(\z[k]) = \exp[-\frac{1}{2}(\z[k]-\muvechat_{\mathbf{z}[k]|s[k]})\transpose\covmathat_{\mathbf{z}[k]|s[k]}^{-1}(\z[k]-\muvechat_{\mathbf{z}[k]|s[k]})] [(2\pi)^\nsensors \det\{\covmathat_{\mathbf{z}[k]|s[k]}\}]^{-1/2}$, since these are readily parametrized by first- and second-order moments and fit well to experimental data~\cite{mcguiness:experimental_results_openair_transmission_macromolecular_communication}.
The proposed approach based on first- and second-order moments can be straightforwardly extended to scenarios with \ac{ISI} by either accounting for the \textit{average} \ac{ISI} contributed by all past symbol sequences or by employing a sequence detector that exploits decisions from pervious symbol intervals to cope with \ac{ISI}.

\scaleSubsection\subsection{Proposed Detectors}\scaleSubsectionBelow\label{sec:detection:detectors}
In Sections~\ref{sec:detection:isi_unaware}-\ref{sec:detection:sequence}, we introduce the different detectors considered in this work before comparing them in Section~\ref{sec:detection:comparison}. Because all detectors leverage approximate moments that are computed using the \ac{UT}, see Section~\ref{sec:detection:ut}, we refer to them as \ac{UT}-based detectors. All proposed detectors require knowledge of the properties of both \ac{TX} and \ac{RX}, i.e., $\ntx(\xbar[k])$, $\nrx(\f{\y[k]})$, and $\f{\y[k]}$. In practice, these can be obtained via \textit{offline} measurements. Additionally, full channel state information is required, i.e., the channel attenuation $\H[\kappa]$ and $\nc(\ybar[k])$, both of which can be estimated via pilot symbols. However, the task of channel estimation is out of scope of this work and will be addressed in future contributions. 

\subsubsection{ISI-unaware Detection}\label{sec:detection:isi_unaware}
The \ac{ISI}-unaware detector, first introduced in the conference version of this paper~\cite{heinlein:nanocom}, assumes that the channel is \ac{ISI}-free, i.e., $\kappamax=0$, resulting in~\eqref{eq:detection:aml} with $\hat{p}_{\z[k]|s[k]}(\z[k])$ being parametrized by $\muvechat_{\z[k]|s[k]}$ and $\covmathat_{\z[k]|s[k]}$, which will be derived in Section~\ref{sec:detection:moments}. 

\subsubsection{Low-Complexity Symbol-by-Symbol Detection}\label{sec:detection:lc}
As we will show in Section~\ref{sec:evaluation:detector}, the \ac{ISI}-unaware detector makes systematic errors in the presence of channel \ac{ISI}. Therefore, we introduce a \ac{LC} symbol-by-symbol detector which exploits statistical information about the \ac{ISI} by averaging over all possible transmitted sequences in the channel memory, resulting in the following decision rule for the symbol estimate $\shatlc[k]$:
\begin{equation}\label{eq:detection:lc}
    \shatlc[k] = \arg\max_{s[k] \in \mathcal{S}_0} \hat{p}_{\z[k]|\mathcal{S}(s[k])}(\z[k]).
\end{equation}
Here, $\mathcal{S}(s[k])$ denotes all sequences $\s[k] = [s[k-\kappamax], \dots, s[k]]$ that contain $s[k]$ and the approximate likelihood is parametrized by $\muvechat_{\z[k]|\mathcal{S}(s[k])}$ and $\covmathat_{\z[k]|\mathcal{S}(s[k])}$, which will be derived in Section~\ref{sec:detection:moments}. While the decisions based on~\eqref{eq:detection:lc} have the same complexity as those for the \ac{ISI}-unaware detector,~\eqref{eq:detection:aml} the mean and covariance are adapted to account for \ac{ISI}.

\subsubsection{Sequence Detection}\label{sec:detection:sequence}
Finally, we introduce the sequence detector, summarized in Algorithm~\ref{alg:detection:sequence}, which achieves much lower \acp{SER} compared to the \ac{LC} detector by incorporating knowledge from pervious intervals. This detector essentially corresponds to the forward pass of the Bahl-Cocke-Jelinek-Raviv (BCJR) algorithm~\cite{bahl:optimal_decoding_linear_codes_minimizing_ser}\footnote{While additional performance gains can be achieved by implementing the full BCJR algorithm (i.e., including the backward pass), this would make symbol decisions available only after the whole sequence has been decoded. Since this would dramatically increase memory requirements, we limit ourselves to the forward pass.}. In the following, we introduce the individual steps of the proposed sequence detector.

\begin{algorithm}[t]
\caption{Sequence Detector\label{alg:detection:sequence}}
\begin{algorithmic}[1]
\State \textbf{Initialization}
\State \text{Initialize state metrics:} $\alpha_{\s[0]} = 1/ |\mathcal{S}_0|^{\kappamax}$

\State \textbf{Iteration}
\For{$k\in \mathbb{N}_{\geq 0}$}
    \State \text{Compute Likelihoods: } $\hat{p}_{\z[k]|\s[k]}(\z[k])$, $\forall \s[k]$
    \State \text{State metric: } $\alpha'_{\s[k]} = \hat{p}_{\z[k]|\s[k]}(\z[k]) \cdot \sum_{\s[k-1] \in \mathcal{P}(\s[k])} \alpha_{\s[k-1]}$
    \State \text{Normalize metric: } $\alpha_{\s[k]} = \alpha'_{\s[k]} / \sum_{\s[k]} \alpha'_{\s[k]}$
    \State \text{Symbol metric: } $\beta_{s[k]} = \sum_{\s[k] \in \mathcal{S}(s[k])} \alpha_{\s[k]}$
    \State \text{Estimate: } $\hat{s}^{\mathrm{seq}}[k] = \arg\max_{s[k] \in \mathcal{S}_0} \beta_{s[k]}$
\EndFor
\end{algorithmic}
\end{algorithm}

The algorithm considers each sequence $\s[k] = [s[k], \dots, s[k-\kappamax]]$ as a \textit{state} to which a state metric $\alpha_{\s[k]}$ is assigned. At initialization, all state metrics are set to identical values following line 2 of the algorithm because there is no prior information available about previously transmitted symbols.
Then, we compute in each interval first the approximate likelihood $\hat{p}_{\z[k]|\s[k]}(\z[k])$ of each state $\s[k]$ for the current observation $\z[k]$ (line 5), as will be shown in Section~\ref{sec:detection:moments}.

Then, we combine these likelihoods with the information from the previous interval where we already computed the metrics $\alpha_{\s[k-1]}$ for each state $\s[k-1]$ to get the unnormalized state metrics $\alpha'_{\s[k]}$ for the current interval (line 6) according to
\vspace*{-2mm}
\begin{equation}
    \alpha'_{\s[k]} = \hat{p}_{\z[k]|\s[k]}(\z[k]) \cdot \sum_{\s[k-1] \in \mathcal{P}(\s[k])} \alpha_{\s[k-1]},
\end{equation}
where $\mathcal{P}(\s[k])$ denotes the set of all possible sequences $\s[k-1]$ that can precede sequence $\s[k]$. The metrics $\alpha'_{\s[k]}$ are then normalized to obtain $\alpha_{\s[k]}$ according to 
\vspace*{-2mm}
\begin{equation}
    \alpha_{\s[k]} = \alpha'_{\s[k]} / \sum_{\s'[k] \in \mathcal{S}} \alpha'_{\s'[k]},
\end{equation}
where $\mathcal{S}$ denotes the set of all possible sequences $ [s[k], \dots, s[k-\kappamax]]$.
With the new state metrics in place, we can compute the symbol metrics $\beta_{s[k]}$ by marginalizing over the normalized state metrics (line 8). Then, we estimate the symbol as $\hat{s}^{\mathrm{seq}}[k]$ by selecting the symbol with the highest symbol metric (line 9). 

\subsubsection{Comparison}\label{sec:detection:comparison}

After introducing the different detectors, we now provide a high-level comparison of their relative performance in terms of their \textbf{error rates} for a given \ac{SNR} and computational complexity. In the \ac{ISI}-free case, all proposed detectors collapse to the \ac{ISI}-unaware detector resulting in identical \ac{SER} performance.
However, in settings with \ac{ISI}, the \ac{ISI}-unaware detector makes systematic decision errors as the \ac{ISI} shifts the signal means, resulting in a mismatch between its decision regions and the actual sensor outputs, and consequently a high \ac{SER} independent of the \ac{SNR}. While the \ac{LC} alleviates this and achieves a lower \ac{SER} for sufficient \acp{SNR}, the sequence detector achieves the lowest \ac{SER} among all proposed detectors for a given \ac{SNR}.
In terms of \textbf{computational complexity}, the \ac{ISI}-unaware and the \ac{LC} detector are identical and require only one likelihood computation per symbol $s[k] \in \mathcal{S}_0$, i.e., their complexity scales linearly with $|\mathcal{S}_0|$. On the other hand, the sequence detector requires one likelihood computation per sequence $\s[k] \in \mathcal{S}$, with $|\mathcal{S}| = |\mathcal{S}_0|^{\kappamax}$. Additionally, the sequence detector has to store $|\mathcal{S}|$ states and $|\mathcal{S}|$ likelihoods, resulting in a larger memory footprint compared to the \ac{ISI}-unaware and \ac{LC} detectors that need to store only $|\mathcal{S}_0|$ states. 
If the \ac{RX} is computationally constrained but computational power is available at the \ac{TX}, the \ac{LC} detector can be combined with an adaptive transmission scheme, which we will introduce in Section~\ref{sec:mixture:adaptive}, to achieve similar \acp{SER} like the sequence detector.

\scaleSubsection\subsection{Exploiting the Unscented Transform}\scaleSubsectionBelow\label{sec:detection:ut}
Before deriving the moments for the proposed detectors in Section~\ref{sec:detection:moments}, we briefly introduce the \ac{UT}~\cite{julier:unscented_filtering_nonlinear_estimation} and explain how to leverage it to \textbf{propagate moments through non-linear functions} and to estimate moments of complex noise distributions. The \ac{UT} is used to estimate the mean and covariance of an $\nsensors$-dimensional random vector $\z$ which is obtained by feeding an $\nspecies$-dimensional random vector $\y$ with known mean and covariance through a non-linear function $\f\cdot$, i.e., $\z = \f\y$. The main idea behind the UT, as originally introduced for state estimation \cite{julier:unscented_filtering_nonlinear_estimation}, is relatively simple: 
One chooses $\nsigmapoints$ carefully selected points $\y_{\!\sigma,i}$, $i=1, \dots, \nsigmapoints$, so-called sigma points, which match the mean $\muvec_\y$ and covariance $\covmat_\y$ of $\y$. These points $\y_{\!\sigma,i}$ are fed through $\f\cdot$ to obtain $\f{\y_{\!\sigma,i}}$. Then, the mean $\muvec_\z$ and covariance $\covmat_\z$ of $\z$ are estimated as 
\begin{subequations}
    \begin{align}
        \muvechat_\z &= \EUT{\z|\muvec_{y},\covmat_{\y}} = \frac{1}{\nsigmapoints}\sum_{i=1}^{\nsigmapoints} \f{\y_{\!\sigma,i}} \\ %
        \covmathat_{\z} &= \CovUT{\z|\muvec_{y},\covmat_{\y}} = \frac{1}{\nsigmapoints}\sum_{i=1}^{\nsigmapoints} \left( \f{\y_{\!\sigma,i}}-\muvechat_\z \right)\left( \f{\y_{\!\sigma,i}}-\muvechat_\z \right)\transpose. %
    \end{align}
\end{subequations}
Here, $\EUT{\cdot|\cdot}$ and $\CovUT{\cdot|\cdot}$ are the operators for the \ac{UT}-based estimate of the mean and covariance of $\z$, respectively. 
We follow the approach in \cite{julier:unscented_filtering_nonlinear_estimation} and select $\nsigmapoints = 2\nspecies$ sigma points as $\y_{\!\sigma,j} = \muvec_\y + [ (\nspecies \covmat_\y )^{1/2} ]_j$ and $\y_{\!\sigma,j+\nspecies} = \muvec_\y - [ (\nspecies \covmat_\y)^{1/2} ]_j$, where $j = 1, \dots, \nspecies$, and $\left(\nspecies \covmat_\y\right)^{1/2}$ is obtained using the Cholesky decomposition of $\nspecies \covmat_\y$. Note that this is conceptually similar to the idea behind particle filtering, but the \ac{UT} employs only a few carefully selected points~\cite{julier:unscented_filtering_nonlinear_estimation}. 

\begin{remark}
Besides using the \ac{UT} to propagate moments through $\f\cdot$, we also can use it to estimate the moments of noise $\n(\x)$ if the signal $\x$ is described by its mean $\muvec_\x$ and covariance $\covmat_\x$ and we can write $\n(\x)$ as $\n(\x) = \g{\x,\nbase}$. Here, $\g\cdot$ is a multi-dimensional, possibly non-linear function, and $\nbase$ is an appropriately sized vector of zero-mean, unit-variance identically and independently distributed Gaussian \acp{RV}. Then, we can collect $\x$ and $\nbase$ in a vector $\x' = [\x, \nbase]\transpose$ with mean $[\muvec_{\x}, \nullmatrix]\transpose$ and covariance matrix $[\covmat_{\x}, \nullmatrix; \nullmatrix, \I]$ and apply the \ac{UT} to $\g{\cdot,\cdot}$ as described before. Formally, we write the \ac{UT} for mean and covariance estimation of $\n(\x)$ for a given $\muvec_{\x}$ and $\covmat_{\x}$ as $\EUT{\n|\muvec_{\x}, \covmat_{\x}}$ and $\CovUT{\n|\muvec_{\x}, \covmat_{\x}}$, respectively.
\end{remark}

\scaleSubsection\subsection{Derivation of the Moments}\scaleSubsectionBelow\label{sec:detection:moments}
Next, we show how to propagate the first- and second-order moments through the system model before summarizing how the individual propagation steps can be combined to obtain the final moments required by the proposed detectors. 
We start by deriving the moments of $\x[k]$ for a given symbol $s[k]$. 
\begin{proposition}[Moments of $\x$]\label{prop:detection:moments:x}
The first- and second-order moments of $\x[k]$ for a memory-free \ac{TX} can be expressed as:
\vspace*{-4mm}
\begin{subequations}
    \begin{equation}\label{eq:detection:moments:x:general:mu}
      \muvec_{\x[k]|s[k]} = \xbar[k] + \int \ntx \cdot p_{\ntx|\xbar[k]}(\ntx) \mathrm{d}\ntx
    \end{equation}    
    \begin{equation}\label{eq:detection:moments:x:general:C}
      \covmat_{\x[k]|s[k]} = \int \mathbf{c}(\ntx, \muvec_{\ntx|\xbar[k]}) p_{\ntx|\xbar[k]}(\ntx) \mathrm{d}\ntx,
    \end{equation}
    \end{subequations}
    where $\mathbf{c}(\boldsymbol{\alpha},\boldsymbol{\beta}) = (\boldsymbol{\alpha}-\boldsymbol{\beta})(\boldsymbol{\alpha}-\boldsymbol{\beta})\transpose$ with random vector $\boldsymbol{\alpha}$ and vector $\boldsymbol{\beta}$, $p_{\ntx|\xbar[k]}(\ntx)$ is the \ac{pdf} of $\ntx(\xbar)$, and $\muvec_{\ntx|\xbar[k]}$ is the conditional mean of $\ntx$ for a given $\xbar$.
\end{proposition}
\begin{proof}[Proof for \eqref{eq:detection:moments:x:general:mu}] 
We first apply the expectation operator to \eqref{eq:system_model:tx_noise} and exploit that the mapping $s[k] \rightarrow \xbar[k]$ is deterministic, yielding
\begin{align*}
    \muvec_{\x[k]|s[k]} = \xbar[k] + \E{\ntx(\xbar[k])|\xbar[k]}.
\end{align*}
Then, we apply the definition of the first-order moment ($\E{x} = \int x \cdot p_x(x) \mathrm{d}x$ for \ac{RV} $x$) to arrive at~\eqref{eq:detection:moments:x:general:mu}.
\end{proof}
\begin{proof}[Proof for \eqref{eq:detection:moments:x:general:C}] The proof is analogous to the proof for \eqref{eq:detection:moments:x:general:mu}, only that the definition of the covariance is inserted into the following expression:
    \begin{align*}
            \covmat_{\x[k]|s[k]} = \Cov{\xbar[k] + \ntx(\xbar[k]) |s[k]} = \underbrace{\Cov{\xbar[k]|s[k]}}_{=\nullmatrix} + \Cov{\ntx(\xbar[k])|s[k]}. \;\;\;\; \makebox[0pt][l]{\qedhere}
        \end{align*}
\end{proof}

\begin{corollary}[\ac{SIN}, \ac{SDCN}] For both \ac{SIN} and \ac{SDCN}, the \ac{TX} noise is signal-independent and thus, we have
\vspace*{-3mm}
\begin{subequations}
    \begin{equation}\label{eq:detection:moments:x:sinsdcn:mu}
      \muvec_{\x[k]|s[k]} = \xbar[k] + \muvec_{\ntx}^{\SIN}
    \end{equation}    
    \begin{equation}\label{eq:detection:moments:x:sinsdcn:C}
      \covmat_{\x[k]|s[k]} = \covmat_{\ntx}^{\SIN}.
    \end{equation}
    \end{subequations}
\end{corollary}
\begin{proof}Apply $\E{\cdot}$ and $\Cov\cdot$ to \eqref{eq:system_model:tx_noise} and then exploit that the symbol-independent \ac{TX} noise is parametrized by $\muvec_{\ntx}^{\SIN}$ and $\covmat_{\ntx}^{\SIN}$.
\end{proof}
\begin{proposition}[Moments of $\ybar$ for a given sequence]\label{prop:detection:moments:yprime:sequence}The mean and covariance of $\ybar[k]$ for a given symbol sequence $\s[k]$ can be written in terms of the \ac{TX} noise as follows:
\vspace*{-2mm}
\begin{subequations}
    \begin{equation}\label{eq:detection:moments:yprime:seq:general:mu}
        \muvec_{\ybar[k]|\s[k]} = \sum_{\kappa=0}^{\kappamax} \H[\kappa] \muvec_{\x[k-\kappa]|s[k-\kappa]}
    \end{equation}
    \begin{equation}\label{eq:detection:moments:yprime:seq:general:C}
        \covmat_{\ybar[k]|\s[k]} = \sum_{\kappa=0}^{\kappamax} \H[\kappa] \covmat_{\x[k-\kappa]|s[k-\kappa]} \H\transpose[\kappa].
    \end{equation}
\end{subequations}
\end{proposition}

\begin{proof}[Proof for \eqref{eq:detection:moments:yprime:seq:general:mu}]We exploit the linearity of the expectation operator together with the memory-free property of the \ac{TX}, i.e., that $\x[k]$ depends only on $s[k]$ and not on previous symbols:
\vspace*{-4mm}
\begin{align*}
    \muvec_{\ybar[k]|\s[k]} &= \E{ \sum_{\kappa=0}^{\kappamax} \H[\kappa] \x[k-\kappa] |\s[k]} \\
    &= \sum_{\kappa=0}^{\kappamax} \H[\kappa] \E{\x[k-\kappa] |s[k-\kappa]} 
    = \sum_{\kappa=0}^{\kappamax} \H[\kappa] \muvec_{\x[k-\kappa]|s[k-\kappa]}.\qedhere
\end{align*}
\end{proof}

\begin{proof}[Proof for \eqref{eq:detection:moments:yprime:seq:general:C}]Analogously, we exploit the memory-free property of the \ac{TX}, that the covariance operator is linear with regard to sums, and that $\Cov{\W\x} = \W \Cov{\x} \W\transpose$:
\begin{align*}
    \covmat_{\ybar[k]|\s[k]}
    &= \Cov{\sum_{\kappa=0}^{\kappamax}\H[\kappa] \x[k-\kappa] |\s[k]} \\
    &= \sum_{\kappa=0}^{\kappamax} \H[\kappa] \Cov{\x[k-\kappa] |s[k-\kappa]} \H\transpose[\kappa] 
    = \sum_{\kappa=0}^{\kappamax} \H[\kappa] \covmat_{\x[k-\kappa]|s[k-\kappa]}\H\transpose[\kappa].\qedhere
\end{align*}
\vspace*{-5mm}\end{proof}

\begin{corollary}[\ac{SIN}, \ac{SDCN}] For \ac{SIN} and \ac{SDCN}, the mean of $\ybar[k]$ can be decomposed into a symbol-dependent contribution and a constant noise contribution while the covariance is entirely symbol-independent:
\vspace*{-6mm}
\begin{subequations}
    \begin{equation}\label{eq:detection:moments:yprime:seq:sinsdcn:mu}
        \muvec_{\ybar[k]|\s[k]}^{\SIN} = \underbrace{\sum_{\kappa=0}^{\kappamax} \H[\kappa] \xbar[k-\kappa]}_{\mathrm{desired}} + \underbrace{\sum_{\kappa=0}^{\kappamax} \H[\kappa] \muvec_{\ntx}^{\SIN}}_{\mathrm{const.\;noise}}
    \end{equation}
    \begin{equation}\label{eq:detection:moments:yprime:seq:sinsdcn:C}
        \covmat_{\ybar[k]|\s[k]}^{\SIN} = \sum_{\kappa=0}^{\kappamax} \H[\kappa] \covmat_{\ntx}^{\SIN} \H\transpose[\kappa].
    \end{equation}
\end{subequations}
\end{corollary}
\begin{proof}
For the mean decomposition, insert $\muvec_{\x[k]|s[k]} = \xbar[k]+\muvec_{\ntx}^{\SIN}$ into~\eqref{eq:detection:moments:yprime:seq:general:mu}. For the covariance, insert $\covmat_{\x[k]|s[k]} = \covmat_{\ntx}^{\SIN}$ into \eqref{eq:detection:moments:yprime:seq:general:C}.
\end{proof}

\begin{proposition}[Moments of $\ybar$ for all possible sequences]\label{prop:detection:moments:yprime:all_sequences}
The mean and covariance of $\ybar[k]$ for symbol $s[k]$ and unknown previous symbols are described by \eqref{eq:detection:moments:yprime:all:general:mu} and \eqref{eq:detection:moments:yprime:all:general:C}, shown at the top of the this page, where $\mathcal{S}(s[k])$ is the set of all sequences $\s[k]$ that contain symbol $s[k]$ and we have $\muvec' = \frac{1}{|\mathcal{S}_0|} \sum_{s \in \mathcal{S}_0} \muvec_{\x|s}$.
\begin{figure*}
    \centering
    \noindent\makebox[\linewidth]{\rule{\linewidth}{0.4pt}}%
    \begin{subequations}
    \begin{equation}\label{eq:detection:moments:yprime:all:general:mu}
        \muvec_{\ybar[k]|\mathcal{S}(s[k])} = \H[0] \muvec_{\x[k]|s[k]} + \sum_{\kappa=1}^{\kappamax} \H[\kappa] \left(\frac{1}{|\mathcal{S}_0|} \sum_{s \in \mathcal{S}_0} \muvec_{\x[k-\kappa]|s} \right)
    \end{equation}
    \begin{equation}\label{eq:detection:moments:yprime:all:general:C}
        \covmat_{\ybar[k]|\mathcal{S}(s[k])} = \H[0] \covmat_{\x[k]|s[k]} \H\transpose[0] + \sum_{\kappa=1}^{\kappamax} \H[\kappa] \left( \frac{1}{|\mathcal{S}_0|} \sum_{s \in \mathcal{S}_0} (\muvec_{\x|s} -\muvec')(\muvec_{\x|s} -\muvec')\transpose + \covmat_{\x|s} \right) \H\transpose[\kappa]
    \end{equation}
    \end{subequations}
    \noindent\makebox[\linewidth]{\rule{\linewidth}{0.4pt}}
\end{figure*}

\end{proposition}
\begin{proof}[Proof for \eqref{eq:detection:moments:yprime:all:general:mu}] Applying the expectation operator to \eqref{eq:system_model:isi} and exploiting its linearity yields
\begin{equation*}
    \muvec_{\ybar[k]|\mathcal{S}(s[k])} = \sum_{\kappa=0}^{\kappamax} \H[\kappa] \E{\x[k-\kappa] | \mathcal{S}(s[k])}.
\end{equation*}
Then, we split the sum into two parts, one for $\kappa=0$ and one for $\kappa > 0$, and exploit that $\E{\cdot|\mathcal{S}(s[k])}=\E{\cdot|s[k]}$ for $\kappa=0$ and $\E{\cdot|\mathcal{S}(s[k])}=\E{\cdot|\mathcal{S}_0}$ for $\kappa > 0$. This yields
\begin{equation*}
    \muvec_{\ybar[k]|\mathcal{S}(s[k])} \!= \!\H[0] \E{\x[k] | s[k]} + \!\sum_{\kappa=1}^{\kappamax}\! \H[\kappa] \E{\x[k-\kappa] | \mathcal{S}_0}.
\end{equation*}
To get the final result, we use $\E{\x[k-\kappa] | \mathcal{S}_0} = \frac{1}{|\mathcal{S}_0|} \sum_{s \in \mathcal{S}_0} \muvec_{\x|s}$ for $\kappa>0$, where $\muvec_{\x|s}$ is the mean of $\x$ for a given symbol $s$.
\end{proof}
\begin{proof}[Proof for \eqref{eq:detection:moments:yprime:all:general:C}] To obtain $\covmat_{\ybar[k]|\mathcal{S}(s[k])}$, we first exploit the linearity of the covariance operator with regard to sums and split the sum into one part for $\kappa = 0$ and one for $\kappa>0$, yielding
\begin{equation*}
    \begin{aligned}
        \covmat_{\ybar[k]|\mathcal{S}(s[k])} &= \H[0] \Cov{\x[k] | s[k]} \H\transpose[0] \\&+ \sum_{\kappa=1}^{\kappamax} \H[\kappa] \Cov{\x[k-\kappa] | \mathcal{S}_0} \H\transpose[\kappa].
    \end{aligned}
\end{equation*}
Then, we can set $\Cov{\x[k] | s[k]} = \covmat_{\x[k]|s[k]}$ and exploit the law of total covariance for \mbox{$\kappa>0$, i.e., }
\begin{equation*}
    \Cov{\x[k-\kappa]|\mathcal{S}_0} = \Cov{\E{\x|s}|\mathcal{S}_0} + \E{\Cov{\x|s}|\mathcal{S}_0},
\end{equation*}
where we dropped the time index for brevity. 
We know that $\E{\x|s} = \muvec_{\x|s}$ and thus $\Cov{\muvec_{\x|s}|\mathcal{S}_0} = \frac{1}{|\mathcal{S}_0|}\sum_{s}(\muvec_{\x|s}-\muvec')(\muvec_{\x|s}-\muvec')\transpose$, where $\muvec'=\frac{1}{|\mathcal{S}_0|} \sum_{s \in \mathcal{S}_0}\muvec_{\x|s}$. Finally, we use that $\E{\Cov{\x|s}|\mathcal{S}_0} = \frac{1}{|\mathcal{S}_0|} \sum_{s \in \mathcal{S}_0} \covmat_{\x|s}$ to obtain \eqref{eq:detection:moments:yprime:all:general:C}. 
\end{proof}

\begin{corollary}[\ac{SIN}, \ac{SDCN}]\label{sec:detection:moments:yprime:sin_sdcn}
For both \ac{SIN} and \ac{SDCN}, the mean can be divided into a contribution from the current symbol and a constant contribution from noise and the previous symbols. The covariance can be divided into one component due to the transmitter noise and into one component due to the lack of knowledge about previously transmitted symbols:
\begin{subequations}
    \begin{equation}\label{eq:detection:moments:yprime:all:sinsdcn:mu}
        \muvec_{\ybar[k]|\mathcal{S}(s[k])}^{\SIN} =  \hspace*{-2mm}\underbrace{\H[0]\xbar}_{\mathrm{symbol\; contr.}}\hspace*{-2mm} +\! \underbrace{\sum_{\kappa=1}^{\kappamax} \H[\kappa] \bar{\muvec}\!\!+ \!\!\sum_{\kappa=0}^{\kappamax} \!\!\H[\kappa] \muvec_{\ntx}^{\SIN}}_{\mathrm{constant\; contr.}}
    \end{equation}
    \begin{equation}\label{eq:detection:moments:yprime:all:sinsdcn:C}
        \begin{aligned}
            \covmat_{\ybar[k]|\mathcal{S}(s[k])}^{\SIN} &= \underbrace{\covmat_{\ybar[k]|\s[k]}^{\SIN}}_{\mathrm{TX\; noise}} + \underbrace{\sum_{\kappa=1}^{\kappamax} \H[\kappa]\! \left(\!\!\frac{1}{|\mathcal{S}_0|}\! \sum_{\xbar \in \symbolalphabet} (\xbar -\bar{\muvec})(\xbar -\bar{\muvec})\transpose \right) \H\transpose[\kappa]}_{\mathrm{symbol\; uncertainty}},
        \end{aligned}
    \end{equation}
\end{subequations}
where $\bar{\muvec} = \frac{1}{|\mathcal{S}_0|} \sum_{\xbar \in \symbolalphabet}\xbar$.
\end{corollary}
\begin{proof}[Proof for \eqref{eq:detection:moments:yprime:all:sinsdcn:mu}]We insert $\muvec_{\x[k]|s[k]} = \xbar[k] + \muvec_{\ntx}^{\SIN}$ into \eqref{eq:detection:moments:yprime:all:general:mu}, exploit that $\frac{1}{|\mathcal{S}_0|} \sum_{s[k] \in \mathcal{S}_0} \muvec_{\x[k]|s[k]}=\bar{\muvec}+\muvec_{\ntx}^{\SIN}$, and split the sums.
\end{proof}
\begin{proof}[Proof for \eqref{eq:detection:moments:yprime:all:sinsdcn:C}]The covariance can be obtained from~\eqref{eq:detection:moments:yprime:all:general:C} by setting $\covmat_{\x|s}=\covmat_{\ntx}^{\SIN}$ and taking the sum over all covariance matrices. For the final expression, we use $\muvec_{\x|s}-\muvec' = \xbar-\bar{\muvec}$.
\end{proof}

\begin{remark}\label{remark:detection:moments:yprime:sin_sdcn} 
    From Corollary~\ref{sec:detection:moments:yprime:sin_sdcn}, it can be seen that i) the uncertainty in this setting is strictly larger than that for a specific sequence $\s[k]$ and ii) the \textit{symbol uncertainty} is the dominating factor in settings with strong \ac{ISI}. The latter becomes apparent when considering that the mixtures are usually spread over the whole feasible set $\feasibleset$, so that the covariance $\sum_{\xbar \in \symbolalphabet} (\xbar -\bar{\muvec})(\xbar -\bar{\muvec})\transpose$ is comparable to the size of $\feasibleset$, whereas $\covmat_{\ntx}^{\SIN}$ is typically much smaller than $\feasibleset$.
\end{remark}

\begin{proposition}[Moments of $\y$]\label{prop:detection:moments:y}The approximate\footnote{Note that we can only obtain approximations in the general case because we have only access to the mean and covariance of $\ybar[k]$ but the full \ac{pdf} of $\ybar[k]$ would be required to obtain the exact mean and covariance of $\nc(\ybar[k])$.} mean and covariance of $\y[k]$ can be determined by combining the results from Proposition~\ref{prop:detection:moments:yprime:sequence} or Propositions~\ref{prop:detection:moments:yprime:all_sequences} with the additional uncertainty introduced by channel noise according to 
\begin{subequations}
    \begin{equation}\label{eq:detection:moments:y:general:mu}
        \muvechat_{\y[k]|C} = \muvec_{\ybar[k]|C} + \EUT{\nc|\muvec_{\ybar[k]|C}, \covmat_{\ybar[k]|C}}
    \end{equation}
    \begin{equation}\label{eq:detection:moments:y:general:C}
        \covmathat_{\y[k]|C} = \covmat_{\ybar[k]|C} + \CovUT{\nc|\muvec_{\ybar[k]|C}, \covmat_{\ybar[k]|C}},
    \end{equation}
\end{subequations}
where $C \in \{\s[k], \mathcal{S}(s[k])\}$. The statistics of $\nc(\ybar[k])$ can be tedious to derive for signal-dependent noise but are easily estimated using the \ac{UT} as described in Section~\ref{sec:detection:ut}. For a solution via integration over conditional \acp{pdf}, we refer the reader to~\cite{heinlein:nanocom}. 
\end{proposition}
\begin{proof}[Proof for \eqref{eq:detection:moments:y:general:mu}]Exploiting the linearity of the expectation operator and~\eqref{eq:system_model:channel_noise} yields
\begin{align*}
    \muvec_{\y[k]|C} &= \E{\ybar[k]|C} + \E{\nc(\ybar[k])|C} \\
    &= \muvec_{\ybar[k]|C} + \E{\nc(\ybar[k])|C}.
\end{align*}
Then, we exploit that $\E{\nc(\ybar[k])|C} \approx \EUT{\nc|\muvec_{\ybar[k]|C},\covmat_{\ybar[k]|C}}$ because $\nc$ depends only on $\ybar[k]$ which in turn depends only on $C$ but there is no direct dependency of $\nc$ on $C$. Here, $\muvec_{\ybar[k]|C}$ and $\covmat_{\ybar[k]|C}$ are obtained from Proposition~\ref{prop:detection:moments:yprime:sequence} or Proposition~\ref{prop:detection:moments:yprime:all_sequences}.
\end{proof}
\begin{proof}[Proof for \eqref{eq:detection:moments:y:general:C}]The proof is analogous to the proof for \eqref{eq:detection:moments:y:general:mu}.
\end{proof}

\begin{corollary}[\ac{SIN}]\label{cor:detection:moments:y:sin}
For \ac{SIN}, we obtain the exact mean and covariance of $\y[k]$ as follows:
    \begin{subequations}
        \begin{equation}\label{eq:detection:moments:y:sin:mu}
            \muvec_{\y[k]|C}^{\SIN} = \muvec_{\ybar[k]|C}^{\SIN} + \muvec_{\nc}^{\SIN}
        \end{equation}
        \begin{equation}\label{eq:detection:moments:y:sin:C}
            \covmat_{\ybar[k]|C}^{\SIN} = \covmat_{\ybar[k]|C}^{\SIN} + \covmat_{\nc}^{\SIN}.
        \end{equation}
    \end{subequations}
\end{corollary}
\begin{proof}
    To obtain \eqref{eq:detection:moments:y:sin:mu}, we insert $\E{\nc(\ybar[k])|C} = \muvec_{\nc}^{\SIN}$ into~\eqref{eq:detection:moments:y:general:mu}. Analogously, we obtain~\eqref{eq:detection:moments:y:sin:C}. 
\end{proof}

\begin{corollary}[\ac{SDCN}] For \ac{SDCN}, we obtain the mean of $\y[k]$ exactly and approximate the covariance using the \ac{UT}:
    \begin{subequations}
        \begin{equation}\label{eq:detection:moments:y:sdcn:mu}
            \muvec_{\y[k]|C}^{\SDCN} = \muvec_{\ybar[k]|C}^{\SIN}
        \end{equation}
        \begin{equation}\label{eq:detection:moments:y:sdcn:C}
            \covmathat_{\y[k]|C}^{\SDCN} = \covmat_{\ybar[k]|C}^{\SIN} + \CovUT{\nc|\muvec_{\ybar[k]|C}, \covmat_{\ybar[k]|C}}.
        \end{equation}
    \end{subequations}
\end{corollary}
\begin{proof}
    To obtain~\eqref{eq:detection:moments:y:sdcn:mu}, we observe that \eqref{eq:system_model:sdcn_definition} is zero-mean according to 
    \begin{align*}
        \E{\left[\sqrt{\nuc \y[k]}\right]_i \cdot[ \nbase]_i} &= \E{\left[\sqrt{\nuc \y[k]}\right]_i} \cdot \E{\left[ \nbase\right]_i} \\
        &= \E{\left[\sqrt{\nuc \y[k]}\right]_i} \cdot 0 = 0,
    \end{align*}
    where $i \in \{1, \dots, \nspecies \}$ and we exploited that the expectation of the product of two independent \acp{RV} is the product of their expectations.
    Eq.~\eqref{eq:detection:moments:y:sdcn:C} is obtained analogously to the proof for \eqref{eq:detection:moments:y:general:mu}. 
\end{proof}

\begin{proposition}[Moments of $\z$]\label{prop:detection:moments:z}Finally, we approximate the first- and second-order moments of $\z[k]$ for a given condition $C \in \{\s[k], \mathcal{S}(s[k])\}$ using the \ac{UT} as follows:
    \begin{subequations}
        \begin{equation}\label{eq:detection:moments:z:general:mu}
            \muvechat_{\z[k]|C} = \muvechat_{\f{\y[k]}|C} + \EUT{\nrx|\muvechat_{\f{\y[k]}|C}, \covmathat_{\f{\y[k]}|C}}
        \end{equation}
        \begin{equation}\label{eq:detection:moments:z:general:C}
            \covmathat_{\z[k]|C} = \covmathat_{\f{\y[k]}|C} + \CovUT{\nrx|\muvechat_{\f{\y[k]}|C}, \covmathat_{\f{\y[k]}|C}}.
        \end{equation}
    \end{subequations}
\end{proposition}
\begin{proof}
    For both $\muvechat_{\z[k]|C}$ and $\covmathat_{\z[k]|C}$, we use the same steps: First, we estimate $\muvechat_{\f{\y[k]}|C}$ and $\covmathat_{\f{\y[k]}|C}$ using the \ac{UT} with the (approximate) mean and covariance from Proposition~\ref{prop:detection:moments:y} or one of the relevant corollaries as input to the \ac{UT}. Then, we continue analogously to the proof of \eqref{eq:detection:moments:y:general:mu} for the final expression.
\end{proof}

\begin{corollary}[\ac{SIN}, \ac{SDCN}]For signal-independent \ac{RX} noise, these expressions simplify to:
    \begin{subequations}
            \begin{equation}\label{eq:detection:moments:z:sinsdcn:mu}
                \muvechat_{\z[k]|C} = \muvechat_{\f{\y}|C} + \muvec_{\nrx}^{\SIN}
            \end{equation}
            \begin{equation}\label{eq:detection:moments:z:sinsdcn:C}
                \covmathat_{\z[k]|C} = \covmathat_{\f{\y}|C} + \covmat_{\nrx}^{\SIN}.
            \end{equation}
        \end{subequations}
\end{corollary}
\begin{proof}
    We use the general formulation for the moments of $\z[k]$ and then continue analogously to the proofs for Corollary~\ref{cor:detection:moments:y:sin}.
\end{proof}

\scaleSubsection\subsection{Obtaining the Final Moments}\scaleSubsectionBelow\label{sec:detection:moments:combination}
Finally, we show how the different propositions, and their respective corollaries, can be combined to obtain the first- and second-order moments utilized by the various proposed detectors. For the \ac{ISI}-unwaware detector, we exploit Propositions~\ref{prop:detection:moments:x}, \ref{prop:detection:moments:yprime:sequence}, \ref{prop:detection:moments:y}, and \ref{prop:detection:moments:z}, while setting $\kappamax=0$. For the \ac{LC} detector, we exploit Propositions~\ref{prop:detection:moments:x}, \ref{prop:detection:moments:yprime:all_sequences}, \ref{prop:detection:moments:y}, and \ref{prop:detection:moments:z} using a $\kappamax$ appropriate for the channel. The sequence detector uses the same propositions like the \ac{LC} detector, only Proposition~\ref{prop:detection:moments:yprime:all_sequences} is replaced by Proposition~\ref{prop:detection:moments:yprime:sequence}. 

%% file: sections/mixture.tex
\scaleSection\section{Moment-based Mixture Alphabet Design}\scaleSectionBelow\label{sec:mixture}
In the previous section, we introduced several detectors for a given mixture alphabet. Now, in Section~\ref{sec:mixture:isi_unaware}, we develop a systematic approach to generate such alphabets, before introducing an adaptive transmission scheme in Section~\ref{sec:mixture:adaptive}.

\scaleSubsection\subsection{Static Mixture Alphabet Design}\scaleSubsectionBelow\label{sec:mixture:isi_unaware}
Existing schemes for designing signal constellations usually made idealizing assumptions about the \ac{RX} characteristics, e.g., the ability to count individual molecules~\cite{wang:effective_constellation_design_CSK}, or consider only a specific type of \ac{RX}~\cite{jamali:olfaction_inspired_MC}. In contrast, here, we propose an algorithm for mixture design that is applicable to general \acp{RX} with non-linear, cross-reactive, and noisy sensing units. 
In this general setting, the identification of \textit{optimal} alphabets, i.e., alphabets that achieve minimal \acp{SER}, is usually infeasible. Thus, we instead propose a suboptimal, greedy approach, summarized in Algorithm~\ref{alg:mixture:isi_unaware}, that optimizes the pair-wise separability of the symbols in the \ac{SOD}. While the algorithm has been developed for \ac{ISI}-free scenarios, it can also be applied in \ac{ISI} scenarios, particularly in combination with the sequence detector introduced in Section~\ref{sec:detection:sequence}\footnote{Designing a mixture alphabet explicitly for \ac{ISI} scenarios is challenging because all symbols are indirectly coupled through their contribution to the \ac{ISI}. However, the proposed algorithm finds a mixture alphabet that is also suitable for \ac{ISI} channels.}.

\begin{algorithm}
\caption{Static Mixture Alphabet Design Algorithm\label{alg:mixture:isi_unaware}}
\begin{algorithmic}[1]
\State \textbf{Initialization}
\State \text{Initialize symbol alphabet:} $\symbolalphabet = \{\}$
\State \text{Generate initial point:} $\mathbf{i}_0 \sim \mathcal{U}(\feasibleset)$
\State \text{Generate candidate set $\mathcal{C}$:} $\mathbf{c} \sim \mathcal{U}(\feasibleset)$, $\forall \mathbf{c} \in \mathcal{C}$
\State \text{Add initial signal point:} $\symbolalphabet = \symbolalphabet \cup\{ \arg\max_{\mathbf{c} \in \mathcal{C}} d(\mathbf{c}, \mathbf{i}_0)$\}

\State \textbf{Iteration}
\While{ $|\symbolalphabet| \leq \nsymbols $ }
    \State \mbox{\text{Add signal point:} $\symbolalphabet = \symbolalphabet \cup\{ \arg\max_{\mathbf{c} \in \mathcal{C}} \;\min_{\xbar \in \symbolalphabet} d(\mathbf{c}, \xbar\}$ }
\EndWhile
\end{algorithmic}
\end{algorithm}
Our approach can be divided into an initialization stage during which the first symbol is chosen and into an iteration stage during which the remaining symbols are added one by one until $\symbolalphabet$ contains the desired number of elements, $\nsymbols$. 
During the \textbf{initialization stage}, we choose an initial point $\mathbf{i}_0$ that is drawn from the uniform distribution over the feasible set $\feasibleset$ (line~3) and we create a candidate set $\mathcal{C}$ whose $|\mathcal{C}|$ elements are uniformly distributed in $\feasibleset$ (line~4). Then, we select the candidate with the largest distance to $\mathbf{i}_0$ as measured by some distance metric $d(\cdot, \cdot)$, which we will introduce below, as the first symbol (line~5). This procedure ensures that the first symbol will usually lie somewhere close to the boundary of $\feasibleset$ and not in the middle of $\feasibleset$, which could impair the overall quality of the mixture alphabet. 
Then, during the \textbf{iteration stage} (lines~6-9), we add one new symbol per iteration until the full mixture alphabet is constructed. In each iteration, the candidate with the largest minimum distance to any existing symbol in $\symbolalphabet$ is selected.
The quality of the constructed mixture alphabet is determined by the size of the candidate set $|\mathcal{C}|$ and the chosen distance metric. Thus, we introduce two distance metrics below.

Both metrics utilize the approximate first- and second-order moments of $\z[k]$ that are computed using Propositions~\ref{prop:detection:moments:x}, \ref{prop:detection:moments:yprime:all_sequences}, \ref{prop:detection:moments:y}, and \ref{prop:detection:moments:z} while setting $\kappamax=0$, i.e., analogous to the \ac{ISI}-unaware detector.  
An obvious choice for $d(\cdot,\cdot)$ is the \textbf{Euclidean distance} between the mean values of the two mixtures $\xbar_i$ and $\xbar_j$ in the \ac{SOD}, i.e.,
\vspace*{-2mm}
 \begin{equation}
    \deuclidean(\xbar_i, \xbar_j) = \lVert \muvechat_{\z|\xbar_i} - \muvechat_{\z|\xbar_j} \rVert_2.
\vspace*{-2mm}
\end{equation}
Alternatively, to account for non-uniform signal spread, we can use the distance metric incorporating the symbols' covariances that was introduced in \cite{jamali:olfaction_inspired_MC} as a \textbf{generalization of the \ac{SNR}} to mixture communications. Specifically, this metric increases with the Euclidean distance between the mean values of the two symbols, but it decreases with increasing variance along the shortest path between the mean values. Formally, the metric is given by
\begin{equation}\label{sec:mixture_design:metrics:vahid}
    \dvahid(\xbar_i, \xbar_j) = \frac{ \lVert \muvechat_{\z|\xbar_i} - \muvechat_{\z_j|\xbar} \rVert_2^2 }{ \mathbf{p}\transpose\covmathat_{\z|\xbar_i}\mathbf{p} + \mathbf{p}\transpose\covmathat_{\z|\xbar_j}\mathbf{p} }
\end{equation}
where $\mathbf{p} = \frac{\muvechat_{\z|\xbar_i} - \muvechat_{\z|\xbar_j}}{\lVert \muvechat_{\z|\xbar_i} - \muvechat_{\z|\xbar_j} \rVert_2}$. While computationally more complex, $\dvahid(\cdot, \cdot)$ promises a higher performance compared to $\deuclidean(\cdot, \cdot)$ due to the consideration of the covariances of both symbols. 

\scaleSubsection\subsection{Adaptive Alphabet Design - Chemical Precoding}\scaleSubsectionBelow\label{sec:mixture:adaptive}
In addition to the algorithm for the design of static mixture alphabets, we also propose an adaptive transmission scheme where the \ac{TX} determines in each interval which concentration vector $\xbar[k]$ to transmit. This scheme assumes that the \ac{RX} employs an \ac{LC} detector trained on a static mixture alphabet $\symbolalphabet$ generated by Algorithm~\ref{alg:mixture:isi_unaware}. Then, the \ac{RX} chooses the mixture $\xbar[k]$ so that the resulting mixture at the \ac{RX} is close to $\muvec_{\y[k]|\mathcal{S}(s[k])}$ for the desired symbol $s[k]$. This corresponds to a form of precoding known in wireless communication.

\begin{algorithm}
\caption{Adaptive Transmission Scheme\label{alg:mixture:adaptive}}
\begin{algorithmic}[1]
\State \textbf{Initialization:}
\State \text{Generate candidate set $\candidateset$: } $\c \sim \mathcal{U}(\feasibleset)$, $\forall \c \in \candidateset$
\State \text{Compute} $\muvec_{\x|\c}$ using Proposition~\ref{prop:detection:moments:x} $\forall \c \in \candidateset$

\State \textbf{Symbol Transmission:}
\For{$k \in \mathbb{N}_{\geq0}$}
    \State \text{Compute ISI: } $\y^{\mathrm{ISI}}[k] = \sum_{\kappa=1}^{\kappamax} \H[\kappa] \muvec_{\x|\c[k-\kappa]}$
    \State \text{Ideal mixture: } $\x^{\mathrm{ideal}}[k] = \H^{-1}[0] \cdot \mathrm{proj}_{\feasibleset}\left(\muvec_{\ybar|\mathcal{S}(s[k])} - \y^{\mathrm{ISI}}[k]\right)$
    \State \text{Select best candidate: } $\c[k] = \arg\min_{\c \in \candidateset} \lVert \muvec_{\x|\c} - \x^{\mathrm{ideal}}[k] \rVert_2$
\EndFor
\end{algorithmic}
\end{algorithm}
Our approach, summarized in Algorithm~\ref{alg:mixture:adaptive}, can be divided into an initialization stage and into the actual symbol transmission. 
During the \textbf{initialization stage}, we create a candidate set $\mathcal{C}$ whose $|\mathcal{C}|$ elements are uniformly distributed in $\feasibleset$ (line 2), and then we compute for each candidate, $\c \in \mathcal{C}$, $\muvec_{\x|\c}$ using Proposition~\ref{prop:detection:moments:x} (line 3). 
During \textbf{symbol transmission}, we adaptively compute the expected \ac{ISI} caused by the previously transmitted mixtures (line 6) according to 
\begin{equation}
    \y^{\mathrm{ISI}}[k] = \sum_{\kappa=1}^{\kappamax} \H[\kappa] \muvec_{\x|\c[k-\kappa]},
\end{equation}
where $\c[k-\kappa]$ denotes the mixture released $\kappa$ intervals ago.
Then, in line 7, we identify which mixture should be transmitted over the channel so that the resulting mixture at the \ac{RX} is closest to the mixture lying in the center of the decision region for $s[k]$ according to 
\begin{equation}
    \x^{\mathrm{ideal}}[k] = \H^{-1}[0] \cdot \mathrm{proj}_{\feasibleset}\left(\muvec_{\ybar|\mathcal{S}(s[k])} - \y^{\mathrm{ISI}}[k]\right),
\end{equation}
where $\mathrm{proj}_{\feasibleset}(\cdot)$ projects a point onto $\feasibleset$. This is necessary to account for the possibility that a non-realizable concentration would be ideal. Then, we instead select the candidate $\c$ whose $\muvec_{\x|\c}$ is closest to $\x^{\mathrm{ideal}}[k]$ (line 8). 
It should be noted that the proposed algorithm is computationally very inexpensive as we do not propagate mixtures through $\f\cdot$ and the \textit{nearest-neighbor search} in line 8 can be efficiently implemented, e.g., via k-d trees~\cite{friedman:algorithm_finding_best_matches,wang:efficient_nearest_neighbor_search_using_dp}.

%% file: sections/evaluation.tex
\scaleSection\section{Performance Evaluation}\scaleSectionBelow\label{sec:evaluation}
In this section, we present simulation results for the different detectors proposed in Section~\ref{sec:detection}, the mixture alphabet design algorithm proposed in Section~\ref{sec:mixture:isi_unaware}, and the adaptive transmission scheme proposed in Section~\ref{sec:mixture:adaptive}. 

\scaleSubsection\subsection{Simulation Setup}\scaleSubsectionBelow\label{sec:evaluation:setup}
We provide simulations for both the \ac{SIN} and \ac{SDCN} scenarios in conjunction with the non-linear model of an array of \ac{MOS} sensors. We parametrize the model for \ac{MOS} sensors in two ways: First, to demonstrate the necessity of explicitly accounting for the non-linear behavior of sensors, we fit the model to measurement data for two commercially available \ac{MOS} sensors (TGS800 and TGS826) responding to mixtures of ethanol and ammonia~\cite{hirobayashi:verification_logarithmic_model}. For our further experiments, on the other hand, we generate artificial sensor data as detailed in Appendix~\ref{sec:appendix:artificial_sensor_data} to allow for a more general investigation into the importance of accounting for the \ac{RX} characteristics, e.g., to investigate how the number of employed sensors affects communication performance. Relevant simulation parameters are reported in Table~\ref{tab:params}. 
To vary the noise power, and thus the \ac{SNR}, we multiply all covariance matrices by a common factor $\nu$ and set $\nuc = \nu$. Note that we show \acp{SER} in result plots as a function $1/\nu$ analogously to showing \acp{SER} over the \ac{SNR} in wireless communications.

{
\renewcommand{\arraystretch}{1.0}
\begin{table}%
\vspace*{-10mm}
    \centering
    \caption{Default system parameters for the simulations.}
    \vspace*{-4mm}
    \def\arraystretch{1.2}
    \begin{tabular}{|l|c|c||l|c|c|} 
        \hline
        Parameter & Value (SIN) & Value (SDCN) & Parameter & Value (SIN) & Value (SDCN)\\ 
        \hline
        $\nsensors$ & \multicolumn{2}{c||}{$3$} & $\nspecies$ & \multicolumn{2}{c|}{$2$} \\ \hline
        $\H[\kappa]$ & \multicolumn{2}{c||}{$[0.01 \cdot \I \;,\; 0.0025 \cdot \I\;,\; 0.0005 \cdot \I]$} & $|\mathcal{C}|$ & \multicolumn{2}{c|}{$500$} \\ \hline

        $\muvec_{\ntx}$ & \multicolumn{2}{c||}{$\nullmatrix$} & $\covmat_{\ntx}$ & \multicolumn{2}{c|}{$10^6 \cdot \I$} \\ \hline
        $\muvec_{\nc}$ & $10 \cdot \boldsymbol{1}$ & - & $\covmat_{\nc}$ & $10 \cdot \I$ & - \\ \hline
        $\muvec_{\nrx}$ & \multicolumn{2}{c||}{$\nullmatrix$} & $\covmat_{\nrx}$ & \multicolumn{2}{c|}{$10^{-13} \cdot \I$} \\ \hline
        $\feasibleset$ (no ISI) & \multicolumn{2}{c||}{ $[0\;,\; 4 \cdot 10^4]^{\nspecies}$ } & $\feasibleset$ (ISI) & \multicolumn{2}{c|}{ $[0\;,\; 3 \cdot 10^4]^{\nspecies}$ } \\ \hline
        $\feasibleset$ (TGS) & \multicolumn{2}{c||}{ $[20 \cdot 10^3, 10 \cdot 10^4] \times [15 \cdot 10^3, 50 \cdot 10^3]$ } & \multicolumn{3}{c|}{-} \\ 
        \hline
    \end{tabular}
    \label{tab:params}
    \vspace*{-4mm}
\end{table}}

\scaleSubsection\subsection{Detector Analysis}\scaleSubsectionBelow\label{sec:evaluation:detector}
In the following, we first motivate the explicit consideration of non-linear, cross-reactive \ac{RX} behavior before evaluating the performance of our proposed detectors quantitatively. 

\subsubsection{ISI-free Scenario}\label{sec:evaluation:detector:isi_free}
\begin{figure*}[h!]
    \centering
    \includegraphics[width=\linewidth]{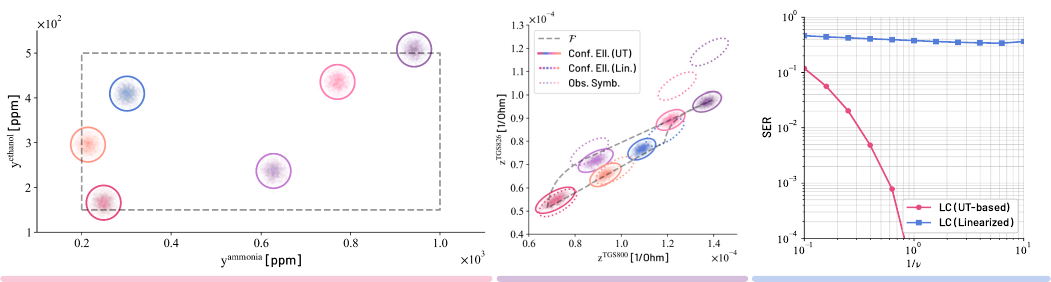}
    \vspace*{-9mm}
    \caption{\textbf{Impact of non-linear, cross-reactive sensor behavior.} The \ac{UT}-based detectors capture the likelihoods of $\y[k]$ (left) and of $\z[k]$ (center). On the other hand, the sensor responses predicted by the linearized sensor models (center, dashed lines) differ systematically from the actual sensor responses (center plot, dots). This is also reflected in the respective \acp{SER} of the \ac{UT}-based detector (right, purple) and the detector relying on a linearized sensor model (right, blue).}
    \label{fig:evaluation:detector:isi_free:non_linearity_motivation}
    \vspace*{-10mm}
\end{figure*}

Before starting the systematic evaluation of the proposed detectors, we first provide some intuition \textbf{why it is important to account for the characteristics of the \ac{RX} array}. To this end, we employ an alphabet with $\nsymbols=6$ symbols optimized for an array of the two considered TGS sensors in the \ac{SIN} scenario\footnote{Here, we use noise parameters deviating from the ones reported in Table~\ref{tab:params} to get more informative results for the chosen sensors. Specifically, we use $\muvec_{\ntx} = \nullmatrix$, $\covmat_{\ntx}$, $\muvec_{\nc} = 10 \cdot \boldsymbol{1}$, $\covmat_{\nc} = 10 \cdot \I$, and $\muvec_{\nrx} = \nullmatrix$, $\covmat_{\nrx} = 10^{-12} \cdot \I$.}. In Figure~\ref{fig:evaluation:detector:isi_free:non_linearity_motivation}, we show on the left and in the middle respectively the empirical likelihoods of $\y[k]$ and $\z[k]$ (colored dots) together with the $3\sigma$-confidence ellipses obtained from the moments used by the \ac{UT}-based detectors (solid lines). For $\z[k]$, we additionally show the $3\sigma$-confidence ellipses obtained from a linearized $\f\cdot$ (dashed lines), i.e., from the first-order two-dimensional Taylor-expansion of $\f\cdot$ around the center of the feasible set. 
Clearly, the moments based on the \ac{UT} match the empirical likelihoods very closely even for $\z[k]$. On the other hand, the moments for the linearized sensor function do not match the empirical likelihoods. This is also reflected in the \acp{SER} of the \ac{UT}-based detector and the detector relying on a linearized sensor model, as shown on the right-hand side of Figure~\ref{fig:evaluation:detector:isi_free:non_linearity_motivation}. While the \ac{SER} of the \ac{UT}-based detector is approx. $10\%$ for $1/\nu = 10^{-1}$ and decreases with increasing $1/\nu$, the detector relying on the linearized system model suffers from \acp{SER} of at least $20\%$ independent of the \ac{SNR}. 

In the following, we analyze the \textbf{quantitative performance} of the detectors proposed in Section~\ref{sec:detection} using artificial sensor arrays that are generated as described in Appendix~\ref{sec:appendix:artificial_sensor_data}. To this end, we compare our \ac{UT}-based detectors to two baselines, namely the \textit{centroid detector} and \textit{\ac{kNN} classifiers}, which we describe in the following in more detail. 

The \textbf{centroid detector} is a \ac{UT}-based detector but with $\covmathat_{\z[k]|s[k]}=\I$, $\forall s[k] \in \mathcal{S}_0$, i.e., only the mean - or centroid - of the likelihoods is considered. While this introduces some sub-optimality compared to the \ac{UT}-based detector which considers the spread of the individual symbols, it has an even lower computational complexity during symbol detection as only the computation of $\ell_2$ distances is necessary. Thus, the centroid detector may be suitable for scenarios with very strict hardware and complexity constraints at the \ac{RX}. 
Furthermore, we consider an \textbf{\ac{kNN} classifier}, which is a commonly used \textit{model-free} approach in machine learning, where new samples are classified according to the classes of the $k$ nearest samples in terms of the Euclidean distance. While this approach has the advantage of not requiring any system model, and instead uses real-world measurements (or simulations), its performance depends on the number of training samples whose generation might be costly if they are obtained from measurements. Furthermore, the complexity during the online phase also increases with the number of training samples, yielding a non-trivial accuracy-complexity trade-off. Thus, we consider \ac{kNN} classifiers trained with four samples per symbol, i.e., the same number of evaluations of $\f\cdot$ during the offline phase as for the \ac{UT}-based detector, and with 100 samples per symbol, respectively.

\begin{figure}
    \centering
    \vspace*{-10mm}
    \includegraphics[width=0.5\linewidth]{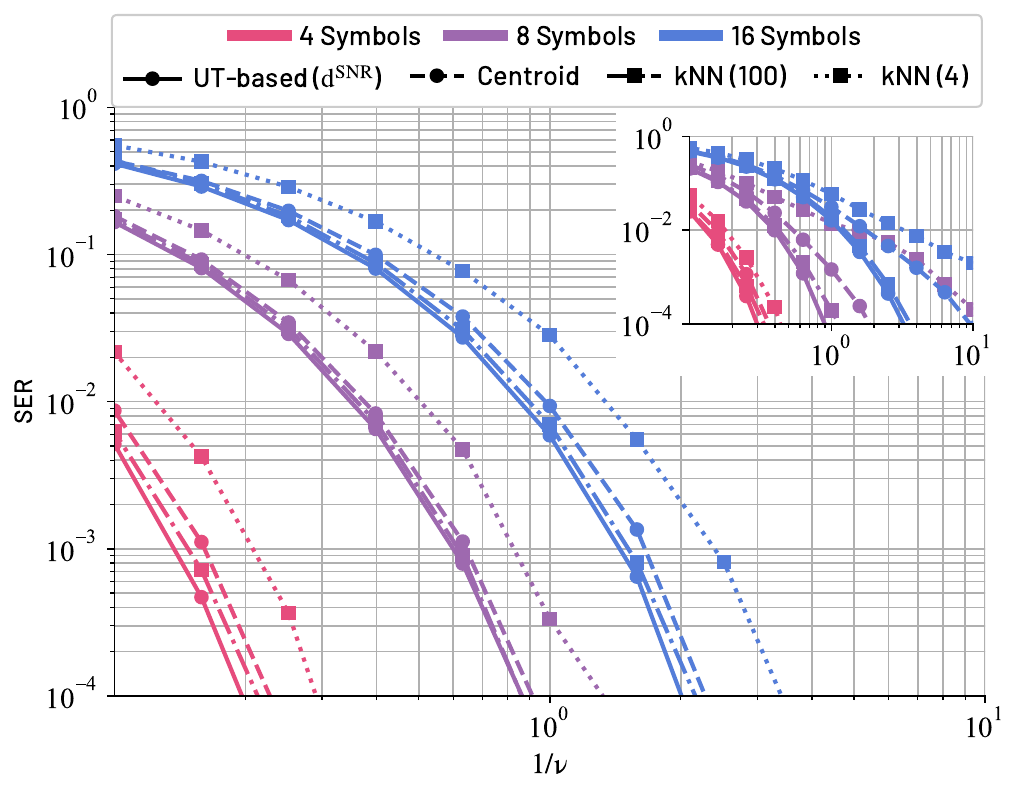}
    \vspace*{-4mm}
    \caption{\textbf{Detector Comparison for \ac{ISI}-free scenarios.} The \ac{UT}-based detector with the covariance-aware distance metric (solid) consistently achieves the lowest \acp{SER} compared to centroid detector (dashed) and \ac{kNN} detectors (dashed-dotted, dotted).}
    \label{fig:evaluation:detector:isi_free:quant}
    \vspace*{-10mm}
\end{figure}
Figure~\ref{fig:evaluation:detector:isi_free:quant} presents the \ac{SER} achieved by the \ac{UT}-based detectors and the aforementioned baselines for 4, 8, and 16 symbols as a function of $1/\nu$ for given mixture alphabets in the \ac{SIN} scenario\footnote{We show the results for the \ac{SDCN} scenario as inset to Figures~\ref{fig:evaluation:detector:isi_free:quant},~\ref{fig:evaluation:detector:isi:quant}, and~\ref{fig:evaluation:mixtures:isi_free:main}. Since similar trends can be observed for both \ac{SIN} and \ac{SDCN}, we focus our discussion on the \ac{SIN} case for brevity and highlight only the differences.}. 
First, we observe that the \ac{SER} decreases as $1/\nu$ increases for all considered detectors. Our proposed \ac{UT}-based detector achieves the lowest \ac{SER} for all $\nu$ although the \ac{kNN} detector with 100 training samples per symbol exhibits almost the same performance at the cost of requiring much more training samples. The \ac{kNN} detector with only four training samples per symbol has significantly worse performance than all other schemes, demonstrating that a model-based approach, such as the proposed \ac{UT}-based detector, can achieve significantly higher \textit{sample efficiency} compared to purely data-driven approaches, i.e., a lower number of training samples is needed to achieve a desired performance. As expected, the centroid detector has some performance loss compared to the \ac{UT}-based detector since it neglects the spread of the likelihoods. Both the centroid detector and the \ac{kNN} detector with four training samples perform worse for \ac{SDCN} than for \ac{SIN}. This is due to the non-uniform signal-spread caused by \ac{SDCN}: The centroid detector does not account for such non-uniform signal spread by design and the \ac{kNN} detector is not able to model the likelihoods appropriately with only four training samples.

In summary, Figure~\ref{fig:evaluation:detector:isi_free:quant} illustrates that the \textbf{\ac{UT}-based detector achieves a lower \ac{SER}} than the two baseline methods, confirming its applicability in settings with significant non-linear and cross-reactive \ac{RX} characteristics. At the same time, the \ac{UT}-based detector requires, in comparison to the data-driven baseline, significantly fewer training samples for a given performance.

\subsubsection{ISI Scenario}\label{sec:evaluation:detector:isi}
In the following, we investigate the performance of the proposed detectors when there is significant \ac{ISI}. 
\begin{figure}
    \centering
    \vspace*{-10mm}
    \includegraphics[width=0.5\linewidth]{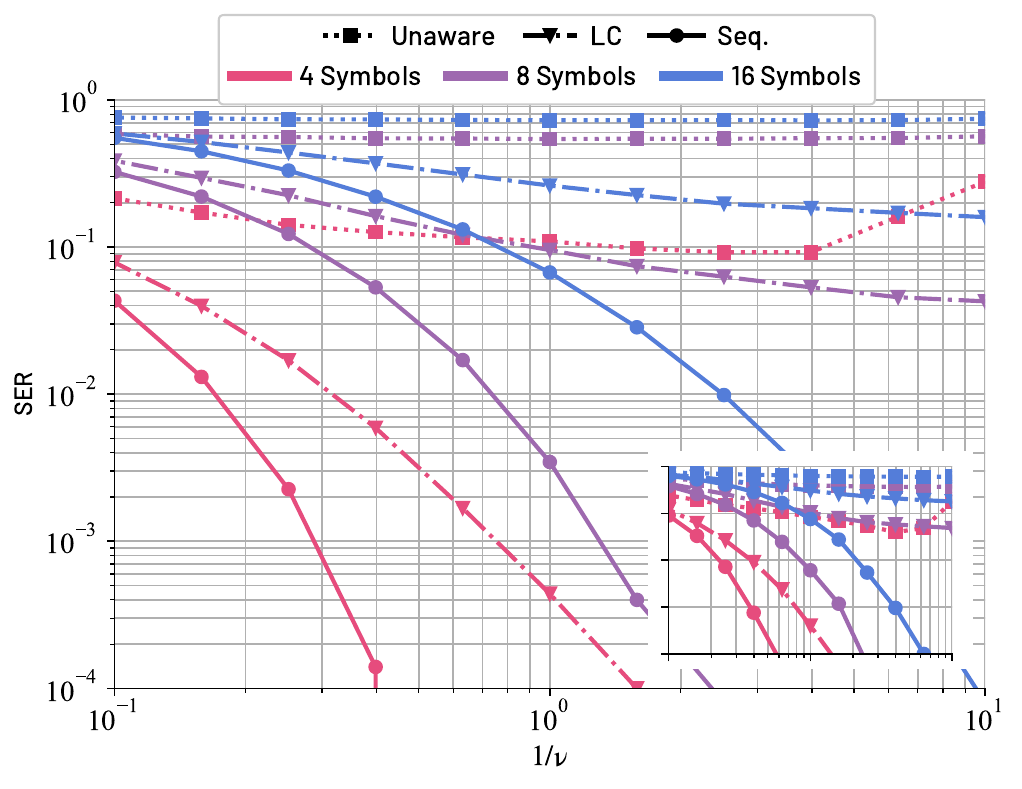}
    \vspace*{-4mm}
    \caption{\textbf{Detector Comparison for \ac{ISI} scenarios.} The \ac{ISI}-unaware detector (dashed) makes systematic decision errors, causing unacceptable \acp{SER} for all \acp{SNR} and alphabet sizes. The \ac{LC} (dashed-dotted) detector improves the performance but still cause high \acp{SER} for higher numbers of symbols, whereas the sequence detector (solid) yields fast \ac{SER} decays with increasing \ac{SNR}.}
    \label{fig:evaluation:detector:isi:quant}
    \vspace*{-10mm}
\end{figure}
Figure~\ref{fig:evaluation:detector:isi:quant} shows the \acp{SER} achieved by the proposed detectors as a function of $1/\nu$. The mixture alphabet optimized for one instance of the artificial sensor array in the \ac{SIN} scenario is adopted. %

First, we observe that the \ac{ISI}-unaware detector (dotted line) yields an unacceptably high \ac{SER} in all cases: Even when $|\mathcal{S}_0|=4$, it barely achieves an \ac{SER} below $10\%$, independent of the \ac{SNR}. This is because $\E{\y[k]}$ shifts compared to the \ac{ISI}-free case due to the additional \ac{ISI} molecules in the channel, causing a systematic mismatch between the expected and the actual sensor outputs\footnote{In the scenario with $4$ symbols, the \ac{SER} of the \ac{ISI}-unaware detector increases again for high \acp{SNR} because the \ac{SNR}-specific detector calibration leads to changing decision regions. For 8 and 16 symbols, the \ac{SER} is, independent of the \ac{SNR}, so high that the effect does not become visible.}. 
On the other hand, the \ac{LC} detector achieves lower \acp{SER} that decrease with increasing \acp{SNR}.
Finally, we observe that the \textbf{sequence detector achieves the lowest \ac{SER}} of all proposed detectors. Its relative advantage over the symbol-by-symbol detectors increases with increasing \ac{SNR} and increasing alphabet size. This matches the insights from Remark~\ref{remark:detection:moments:yprime:sin_sdcn}: The signal spread due to \ac{ISI} is usually much higher than the spread caused by noise. As the \ac{SNR} increases, the spread due to noise decreases while the spread due to \ac{ISI} remains almost constant, explaining the growing advantage of the sequence detector which can exploit information from previous intervals to reduce the uncertainty about the \ac{ISI}.

\scaleSubsection\subsection{Mixture Alphabet Design Analysis}\scaleSubsectionBelow\label{sec:evaluation:mixture}
So far, we have established that accounting for the non-linear, cross-reactive behavior of the \ac{RX} sensor array is important for effective detection. In the following, we focus on the performance of the proposed mixture alphabet design algorithm and the adaptive transmission scheme.

\subsubsection{Mixture Alphabet Design Algorithm}\label{sec:evaluation:mixture:isi_free}
\begin{figure}
    \centering
    \vspace*{-10mm}
    \includegraphics[width=0.5\linewidth]{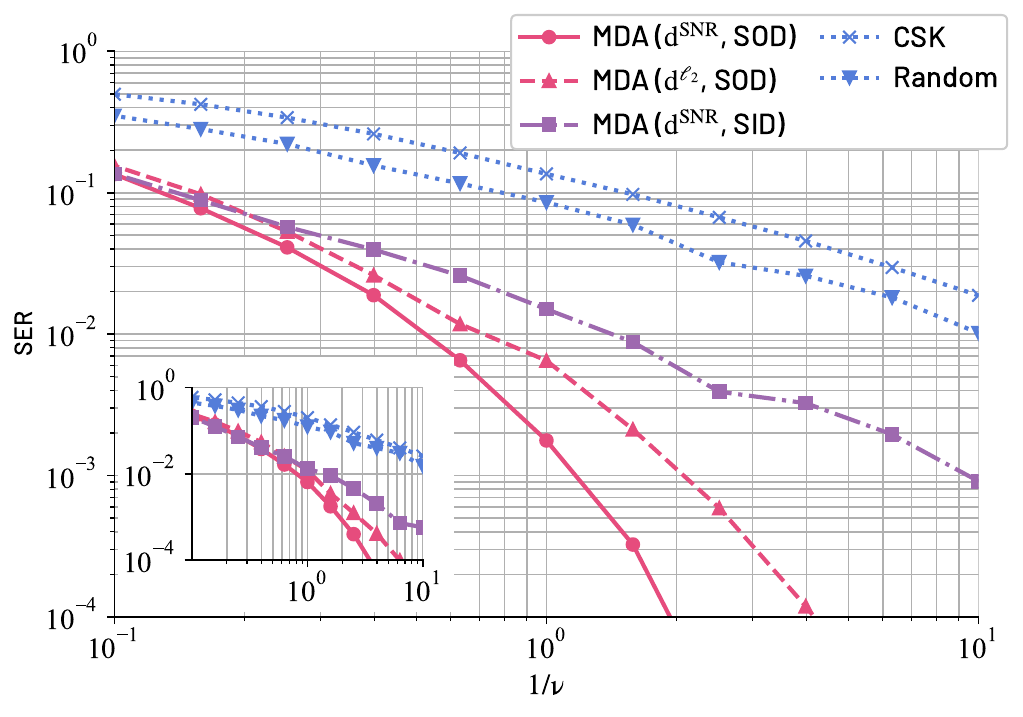}
    \vspace*{-4mm}
    \caption{\textbf{Comparison of Mixture Design Algorithms for \ac{ISI}-free scenarios.} The \ac{SER} of mixture alphabets in the \ac{SOD} with the \ac{SNR}-based distance metric (pink solid) is consistently lower than that of alphabet designs in the \ac{SID} (purple) and baselines like \ac{CSK} or random alphabets (blue).}
    \label{fig:evaluation:mixtures:isi_free:main}
    \vspace*{-10mm}
\end{figure}
In Figure~\ref{fig:evaluation:mixtures:isi_free:main}, we evaluate the mixture alphabets created by the algorithm proposed in Section~\ref{sec:mixture:isi_unaware} for both $\dvahid(\cdot,\cdot)$ and $\deuclidean(\cdot,\cdot)$. To this end, we compare the optimized designs to several baselines for $\nsymbols=8$ with regard to the \ac{SER} achieved by the \ac{UT}-based detectors as performance metric. In our simulations, we average over 100 different \ac{RX} arrays that are generated as described in Appendix~\ref{sec:appendix:artificial_sensor_data} to get generalizable insights. 

As a first baseline, we use \textbf{Algorithm~\ref{alg:mixture:isi_unaware} in the \ac{SID}}, i.e., we compute $\dvahid(\cdot, \cdot)$ for the approximate likelihoods of $\y$ instead of $\z$ and thus do not account for the non-linear, cross-reactive \ac{RX} behavior during mixture design. As a second baseline, we consider \textbf{\ac{CSK} with equally spaced constellation points} for one of the two signaling molecules in $\feasibleset$. This helps to quantify the gain that can be achieved by allowing a second dimension for the mixture design. Finally, we also consider a \textbf{random mixture design}, where the individual symbols are simply drawn independently from $\mathcal{U}(\feasibleset)$.

In Figure~\ref{fig:evaluation:mixtures:isi_free:main}, we compare the proposed \ac{MDA} to the aforementioned baselines for \ac{SIN}. As expected, our \ac{MDA} with $\dvahid(\cdot,\cdot)$ achieves the lowest \ac{SER} for a given value of $1/\nu$. Using $\deuclidean(\cdot,\cdot)$ entails slightly lower computational complexity, but is also subject to performance loss relative to $\dvahid(\cdot,\cdot)$.
Regardless of the employed distance metric, our algorithm outperforms all baselines, including the mixture design in the \ac{SID} which does not account for the non-linear, cross-reactive \ac{RX} behavior\footnote{Note that \ac{CSK} is outperformed here even by random mixture design. This is because we average over 100 \ac{RX} arrays and some of those arrays are not sensitive to the molecule type employed by \ac{CSK}. By exploiting both molecule types, random mixture design is robust against this effect.}. 
However, the advantage of the proposed algorithm compared to optimization in the \ac{SID} is less pronounced for low \acp{SNR}, which we attribute to the fact for low \acp{SNR} that the likelihoods are less accurately described by only the first- and second-order moments when feeding them through $\f\cdot$: As the spread of $\y[k]$ increases, $p_{\z[k]|s[k]}(\z[k])$ may be less accurately described by a Gaussian distribution than for a smaller spread of $\y[k]$, where $\f\cdot$ can be locally approximated as linear. Such a mismatch can lead to scenarios, where optimization in the \ac{SID} can yield similar performance to optimization in the \ac{SOD}.%

\begin{figure}
    \centering
    \vspace*{-8mm}
    \includegraphics[width=0.5\linewidth]{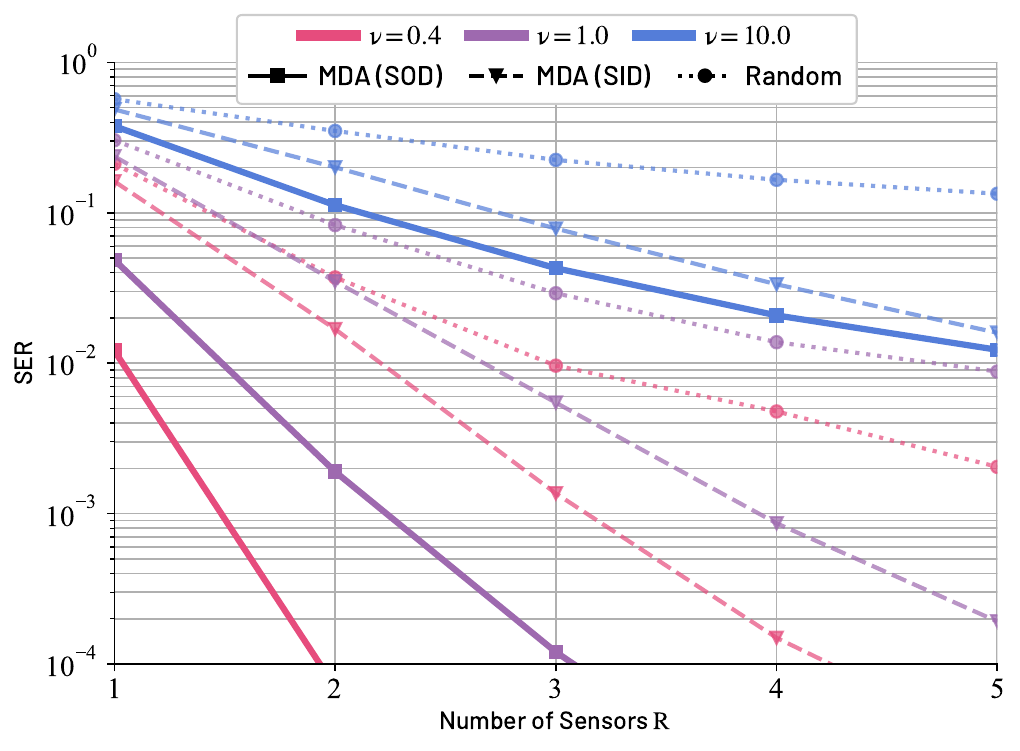}
    \vspace*{-4mm}
    \caption{\textbf{Impact of the optimization domain for different numbers of sensors.} As the number of sensors increases, the \ac{SER} decreases for all alphabet design algorithms and optimization domains. The \ac{SOD}-based optimization (solid) achieves the lowest \ac{SER} compared to \ac{SID}-based optimization (dashed) and random alphabets (dotted).}
    \label{fig:evaluation:mixtures:n_sensors}
    \vspace*{-10mm}
\end{figure}
Finally, we examine how the number of available sensors affects which optimization domain is preferable. We investigate this in Figure~\ref{fig:evaluation:mixtures:n_sensors}, where we use $\nsymbols=8$, $\nspecies=3$, and vary $\nsensors$ from 1 to 5. As before, the sensor arrays are generated as described in Appendix~\ref{sec:appendix:artificial_sensor_data}. All remaining system parameters are kept at their default values for the \ac{SIN} scenario. When comparing the \acp{SER} of the alphabets generated by the proposed \ac{MDA} in the \ac{SOD}, the proposed \ac{MDA} in the \ac{SID}, and randomly, we observe that the \ac{SER} decreases as $\nsensors$ increases for all considered alphabets. This is expected as more independent observations of $\y[k]$ become available as $\nsensors$ increases. We can observe that, in the considered scenario, optimization in the \ac{SOD} can, depending on the \ac{SNR}, achieve the same \ac{SER} as optimization in the \ac{SID} with one or two sensors less.
In summary, Figures~\ref{fig:evaluation:mixtures:isi_free:main} and~\ref{fig:evaluation:mixtures:n_sensors} shows that \textbf{\ac{RX}-aware mixture design can enable lower \acp{SER}} than simply optimizing for separable molecule concentrations.

\subsubsection{ISI Scenario}\label{sec:evaluation:mixture:isi}
We have already shown in Section~\ref{sec:evaluation:detector:isi} that the alphabets generated by the \ac{MDA} enable low \acp{SER} also for channels with \ac{ISI}, especially when combined with the sequence detector proposed in Section~\ref{sec:detection:sequence}. Now, we examine the performance of the \textbf{adaptive transmission scheme} proposed in Section~\ref{sec:mixture:adaptive} and its ability to enable low \acp{SER} when combined with the \ac{LC} symbol-by-symbol detector, thus reducing the required computational complexity at the \ac{RX}.   
\begin{figure}
\vspace*{-10mm}
     \centering
     \begin{subfigure}[b]{0.3\textwidth}
         \centering
         \includegraphics[width=\textwidth]{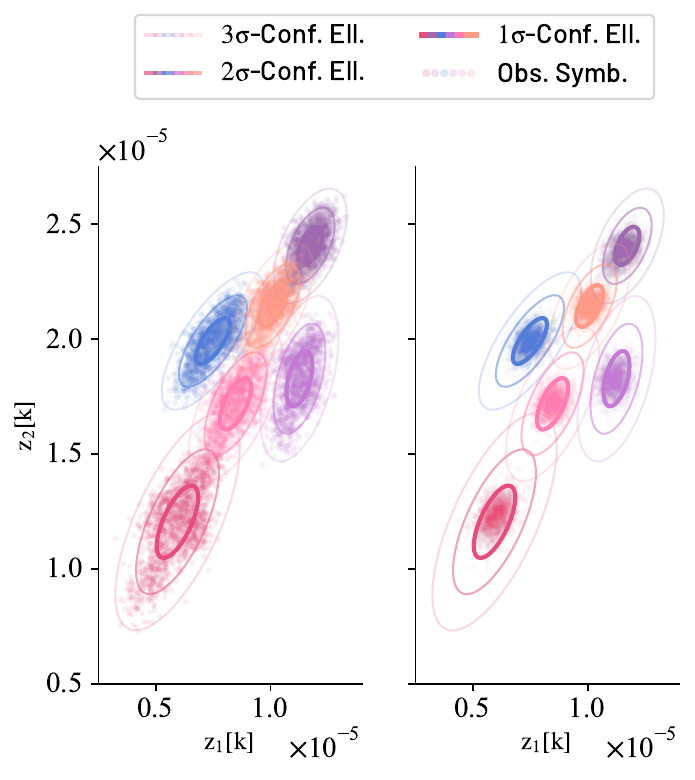}
         \vspace*{-8mm}
         \caption{}
         \label{fig:evaluation:mixtures:adaptive:motivation}
     \end{subfigure}
     \hfill
     \begin{subfigure}[b]{0.5\textwidth}
         \centering
         \includegraphics[width=\textwidth]{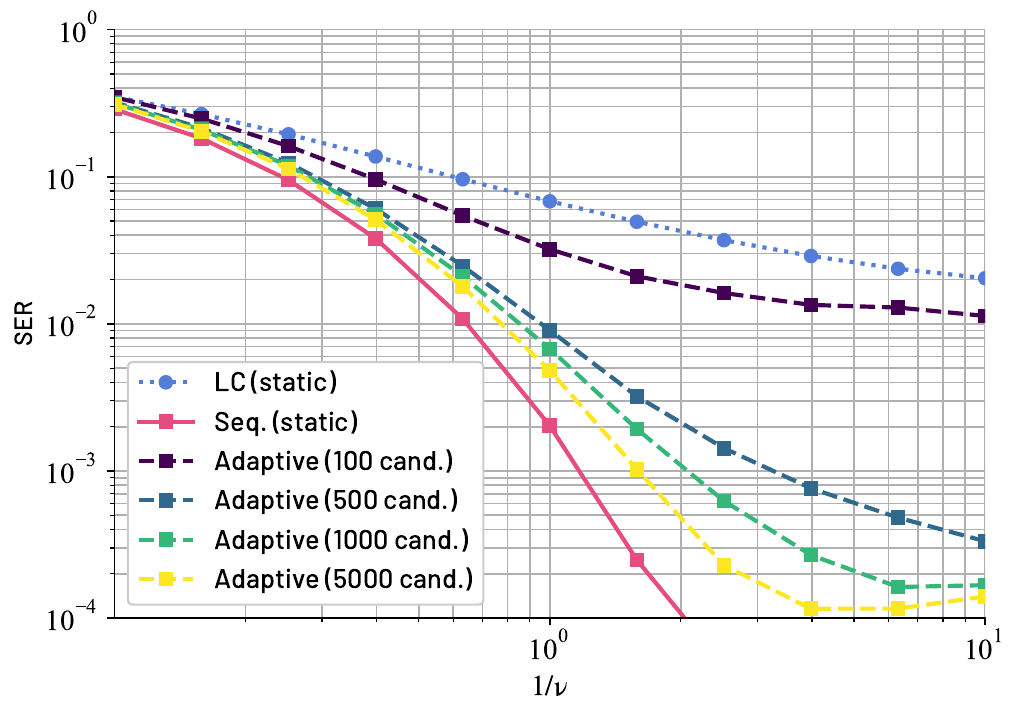}
         \vspace*{-8mm}
         \caption{}
         \label{fig:evaluation:mixtures:adaptive:main}
     \end{subfigure}
     \vspace*{-4mm}
    \caption{\textbf{Adaptive Transmission Scheme.} \textit{(a)} The sensor outputs of the adaptive transmission scheme (right) are more concentrated in the center of the \ac{LC}'s decisions region compared to the static alphabet (left). \textit{(b)} The adaptive scheme (dashed) outperforms the \ac{LC} detector (dotted) and approaches the performance of the sequence \mbox{detector (solid).}}
    \label{fig:evaluation:mixtures:adaptive}
    \vspace*{-10mm}
\end{figure}
To illustrate the working principle of the adaptive scheme, we consider an array with $\nsensors=2$ sensors and an alphabet of size $\nsymbols=6$, generated by the proposed \ac{MDA}. In Figure~\ref{fig:evaluation:mixtures:adaptive:motivation}, we show the sensor outputs for the static alphabet on the left-hand side and the sensor outputs for the adaptive transmission scheme on the right-hand side. Additionally, we show confidence ellipses of the \ac{LC} by solid lines in each plot, where the opacity increases from the $1\sigma$ to the $3\sigma$ confidence ellipse. 
This figure reveals why the adaptive transmission scheme can achieve lower \acp{SER}: The sensor outputs are much more tightly packed in the center of the decision regions when employing the adaptive transmission scheme, enabling a much lower \ac{SER} compared to the static mixture alphabet.
This becomes apparent in Figure~\ref{fig:evaluation:mixtures:adaptive:main}, where we investigate the achieved \ac{SER}. Specifically, we compare the proposed adaptive scheme to a static mixture alphabet, the latter in conjunction with either the \ac{LC} detector (blue dotted) or the sequence detector (red solid).
The dashed lines show the \ac{SER} achieved by the adaptive transmission scheme for different $|\mathcal{C}|$. 
As expected, increasing $|\mathcal{C}|$ decreases the \ac{SER} and closes the gap between the \ac{LC} and the sequence detector for low and moderate \acp{SNR}. Only for very high \acp{SNR}, the adaptive transmission scheme has an error-floor due to the finite $|\mathcal{C}|$ and, depending on the scenario, because the released concentrations are confined to $\feasibleset$. 
In summary, Figure~\ref{fig:evaluation:mixtures:adaptive:main} shows that the proposed \textbf{adaptive transmission scheme enables low \acp{SER} in \ac{ISI} channels} even when deploying \ac{LC} symbol-by-symbol detectors at the \ac{RX}.

%% file: sections/conclusion.tex
\scaleSection\section{Conclusion}\scaleSectionBelow\label{sec:conclusion}
In this work, we studied the detection and alphabet design for molecule mixture communications for \ac{ISI} channels and non-linear, cross-reactive \ac{RX} arrays. Specifically, we proposed a comprehensive system model that accounts for release, propagation, and sensor noise and includes a non-linear, cross-reactive \ac{RX} array.
We then proposed a symbol-by-symbol detector that exploits statistical knowledge of the \ac{ISI} and a sequence detector for static molecule mixture alphabets. Furthermore, we proposed a complementary mixture alphabet design algorithm that explicitly accounts for the employed \ac{RX} and optimizes the separability of the sensor outputs of different symbols. We demonstrated via extensive simulations for \ac{ISI}-free scenarios that our proposed detectors outperform data-driven baseline schemes and that our mixture alphabet design algorithm improves upon schemes that ignore the \ac{RX} characteristics. We also verified that the sequence detector enables effective symbol detection even in scenarios with strong \ac{ISI}. 
If the \ac{RX} has extremely limited computational capacity, an adaptive transmission scheme can be employed in conjunction with a symbol-by-symbol detector to achieve similar performance as the sequence detector. 
Our results highlight that reliable communication is possible with non-linear, cross-reactive sensing units even in the presence of strong \ac{ISI}.
Potential future directions include the experimental validation of the proposed schemes, their extension to incorporate uncertainty regarding system parameters, e.g., due to factors like changes in temperature or humidity, and their extension to multi-user scenarios. 

%% file: sections/appendix.tex
\appendix
\scaleSection\section*{Generation of Artificial Sensor Data}\scaleSectionBelow\label{sec:appendix:artificial_sensor_data}
While the response of \ac{MOS} sensors to individual gases in clean background air is usually reported in the datasheets of those sensors, their response to mixtures is usually not reported in a way that can be used to parametrize \eqref{eq:system_model:reciever_characteristics:f}. Also, sensor parameters are rarely reported in the literature and vary over several orders of magnitude depending on source and system parameters (see, e.g.,~\cite{llobet:steadystate_transient_behavior_of_thickfilm_tin_oxide_sensors_gas_mixtures,gurin:gas_mixture_estimation_using_powerlaw_models_arrayed_chemiresistive_MOS,hirobayashi:verification_logarithmic_model}). 
Due to this lack of reliable data, we generate artificial sensor data that resembles the sensor parameters reported in~\cite{llobet:steadystate_transient_behavior_of_thickfilm_tin_oxide_sensors_gas_mixtures}\footnote{It should be noted that some numerical values reported in \cite[Table~1~\&~2]{llobet:steadystate_transient_behavior_of_thickfilm_tin_oxide_sensors_gas_mixtures} seem inconsistent. For example, parameter $A_i$ for sensor TGS822 is off by a factor of 10 compared to the results reported in Figure~6. Since it is not possible to check all the possible values, we decided to generate artificial data with similar characteristics instead of relying on the individual numbers.}. 
To this end, we perform rejection sampling, i.e., we generate the parameters for $\a$, $\mathbf{A}$, and $\b$ independently and then check whether the joint parameter configuration for the sensor function $\f{\cdot}$ is valid, as will be defined below. Invalid configurations are disregarded and replaced by a new parameter set until a valid configuration is identified. 
A \textbf{\textit{valid} sensor configuration} is monotonously increasing with increasing molecule concentrations throughout $\feasiblesety$. This is evaluated by checking that $[\nabla_{\y} \f{\y}]_i \geq 0$, $\forall i \in \{0, \dots, \nspecies\}$, holds for $10^5$ test points $\y \sim \mathcal{U}(\feasiblesety)$.  
We select the \textbf{sensor parameters} independently for each species according to $\a \sim 10^{n_a}$, $\mathbf{A} \sim 10^{n_A}$, and $\b \sim \mathrm{TruncNorm}(0.5, 0.2, [0.3, 0.7])$. Here, $n_a$ is normal distributed with mean $-6.5$ and standard deviation $0.5$, and $n_A$ is normal distributed with mean $-7.5$ and standard deviation $1$. $\mathrm{TruncNorm}(0.5, 0.2, [0.3, 0.7])$ denotes a normal distribution with mean $0.5$ and standard deviation $0.2$ while being limited to the interval $[0.3, 0.7]$, thereby capturing the range of power law coefficients reported in the literature (see, e.g., \cite{yamazoe:theory_power_laws_semiconducator_gas_sensors,hua:theoretical_investigation_power_law_response_mos,madrolle:linear_quadratic_model_quantification_mixture_two_diluted_gases_single_MOS}).